\documentclass[10pt,reqno,a4paper,oneside,11pt]{amsart}%
\usepackage{tikz}
\usepackage{mathtools}
\usepackage{extarrows}
\usetikzlibrary{arrows}
\usepackage{xparse}
\usepackage{amsfonts}
\usepackage{amsfonts}
\usepackage{amsfonts}
\usepackage{amsfonts}
\usepackage{amsfonts}
\usepackage{amsfonts}
\usepackage{amsfonts}
\usepackage{amsfonts}
\usepackage{amsfonts}
\usepackage{amsfonts}
\usepackage{amsfonts}
\usepackage{amsfonts}
\usepackage{amsfonts}
\usepackage{amsfonts}
\usepackage{amsfonts}
\usepackage{amsfonts}
\usepackage{amsfonts}
\usepackage{amsfonts}
\usepackage{amsfonts}
\usepackage{amsfonts}
\usepackage{mathrsfs}
\usepackage{mathrsfs}
\usepackage{amsfonts}
\usepackage{amssymb}
\usepackage{amsmath}
\usepackage{amsthm}
\usepackage{graphicx}%
\providecommand{\U}[1]{\protect\rule{.1in}{.1in}}
\newtheorem{theorem}{Theorem}[section]

\newtheorem{lemma}[theorem]{Lemma}
\newtheorem{proposition}[theorem]{Proposition}

\newtheorem{definition}{Definition}[section]

\theoremstyle{definition}
\theoremstyle{remark}
\numberwithin{equation}{section}

\ifx\pdfoutput\relax\let\pdfoutput=\undefined\fi
\newcount\msipdfoutput
\ifx\pdfoutput\undefined\else
\ifcase\pdfoutput\else
\ifx\paperwidth\undefined\else
\ifdim\paperheight=0pt\relax\else\pdfpageheight\paperheight\fi
\ifdim\paperwidth=0pt\relax\else\pdfpagewidth\paperwidth\fi
\fi\fi\fi
\begin{document}
\pagestyle{myheadings}

\begin{center}
	{\huge \textbf{The solvability conditions and exact solutions to some  quaternion tensor systems}}\footnote{This research was supported by  the National Natural
		Science Foundation of China (Grant no. 11971294). Email address: wqw@t.shu.edu.cn}

	\bigskip
	
	{ \textbf{Qing-Wen Wang$^{a,*}$, Mengyan Xie$^{a}$}}
	
	{\small
		\vspace{0.25cm}
		
		$a.$ Department of Mathematics, Shanghai University, Shanghai 200444, P. R. China}
\end{center}
\begin{quotation}
\noindent$\mathbf{Abstract}$:
We derive necessary and sufficient conditions for the existence of the exact solution to the  Sylvester-type quaternion tensor system
$
\mathcal{A}_i\ast_{N}\mathcal{X}_i+
\mathcal{Y}_i\ast_{M}\mathcal{B}_i+\mathcal{C}_i\ast_{N}
\mathcal{Z}_i\ast_{M}\mathcal{D}_i+\mathcal{F}_i\ast_{N}
\mathcal{Z}_{i+1}\ast_{M}\mathcal{G}_i=\mathcal{E}_i, i=\overline{1,3}
$
using Moore-Penrose inverse, and present an expression
of the general solution to the system when it is solvable.
As an application of this system, we provide the solvability conditions and general solutions for the Sylvester-type quaternion tensor system
$
\mathcal{A}_i\ast_{N}\mathcal{Z}_i\ast_{M}\mathcal{B}_i+
\mathcal{C}_i\ast_{N}\mathcal{Z}_{i+1}\ast_{M}\mathcal{D}_i=
\mathcal{E}_i, i=\overline{1,4}.
$
This paper can also serve as  extensions to some known results.
\newline%
\end{quotation}

\section{Introduction}
About two hundred years ago the skew fields of quaternions was introduced by Sir William Rowan Hamilton\cite{wrh}, which has all properties of fields except for  the commutativity of multiplication. The theoretical knowledge  about quaternions is a good part of algebra; see, for example, \cite{lr}. Dating back to 1936 \cite{law}, the literature on matrices over quaternions is still fragmentary. Nevertheless, renewed interest and applications have been witnessed recently, such as signal and color image procession, quaternionic quantum mechanics (qQM) and many other fields  (eg. \cite{nlb}, \cite{jhc},  \cite{zhh2},  \cite{sdl},   \cite{cct1}-- \cite{qww7}, \cite{wms}, \cite{sfy},  \cite{fzz}).

Since 1952, the first generalized Sylvester matrix equation was studied by Roth. This seems to have stimulated several authors who have discussed the generalized Sylvester matrix equations and their applications, for instance, image processing \cite{sk}, eigenvalue assignment problems \cite{sb}, neural network \cite{ynz} and so on.
It is known that matrix version is the special form of tensor version. Tensor theory  has arisen lots of applications, such as signal processing (\cite{ldl1}--\cite{ldl4}), graphic analysis \cite{tgk1} and so on (eg. \cite{bwb}, \cite{rb1}, \cite{ckc}, \cite{dwy1}, \cite{dwy2}, \cite{gy1}--\cite{gy3}, \cite{zhh}, \cite{tgk}, \cite{lml}, \cite{qlq}--\cite{qlq2}, \cite{sjy}, \cite{sjy2}, \cite{qww8}, \cite{wym}).

There are also some papers about the generalized Sylvester tensor equation (eg. \cite{bm}, \cite{cz},  \cite{zhh}, \cite{wml}, \cite{ll1}, \cite{ll2}, \cite{sxh}, \cite{slz}, \cite{qww8}). For example,
Sun et al.  \cite{slz} derived the solvability conditions and solutions for the Sylvester tensor equation
\[\mathcal{A}\ast_{N}\mathcal{X}+
\mathcal{X}\ast_{M}\mathcal{D}=\mathcal{C},\]
where $\mathcal{A}\in\mathbb{R}^{I_{1}\times\cdots\times I_{N}\times I_{1}\times\cdots\times I_{N}},~\mathcal{D}\in\mathbb{R}^{I_{1}\times\cdots\times I_{N}\times J_{1}\times\cdots\times J_{M}},~\mathcal{C}\in\mathbb{R}^{J_{1}\times\cdots\times J_{M}\times J_{1}\times\cdots\times J_{M}}$. The solvability conditions and solutions of two-sided Sylvester tensor equation
\[\mathcal{A}\ast_{N}\mathcal{X}_1\ast_{M}\mathcal{B}+
\mathcal{C}\ast_{N}\mathcal{X}_{2}\ast_{M}\mathcal{D}=
\mathcal{E}\]
was studied by He \cite{zhh}, where $\mathcal{A}\in\mathbb{H}^{I_{1}\times\cdots\times I_{N}\times Q_{1}\times\cdots\times Q_{N}},$ $\mathcal{B}\in\mathbb{H}^{P_{1}\times\cdots\times P_{M}\times J_{1}\times\cdots\times J_{M}},$
$\mathcal{C}\in\mathbb{H}^{I_{1}\times\cdots\times I_{N}\times L_{1}\times\cdots\times L_{N}},$ $\mathcal{D}\in\mathbb{H}^{S_{1}\times\cdots\times S_{M}\times J_{1}\times\cdots\times J_{M}},$
$\mathcal{E}\in\mathbb{H}^{I_{1}\times\cdots\times I_{N}\times J_{1}\times\cdots\times J_{N}}$.
He et al. \cite{zhh1} gave some solvability conditions for the Sylvester quaternion tensor system
\begin{equation}\label{mainsystem00}
\begin{cases}
\mathcal{A}_1\ast_{N}\mathcal{X}_1-
\mathcal{X}_{2}\ast_{N}\mathcal{D}_1=
\mathcal{E}_1\\
\mathcal{A}_2\ast_{N}\mathcal{X}_2-
\mathcal{X}_{3}\ast_{N}\mathcal{D}_2=
\mathcal{E}_2\\
\mathcal{A}_3\ast_{N}\mathcal{X}_3-
\mathcal{X}_{4}\ast_{N}\mathcal{D}_3=
\mathcal{E}_3\\
\mathcal{A}_4\ast_{N}\mathcal{X}_4-
\mathcal{X}_{5}\ast_{N}\mathcal{D}_4=
\mathcal{E}_4,
\end{cases}
\end{equation}
where $\mathcal{A}_i\in\mathbb{H}^{I_{1}\times\cdots\times I_{N}\times Q_{1}\times\cdots\times Q_{N}},$ 
 $\mathcal{D}_i\in\mathbb{H}^{S_{1}\times\cdots\times S_{M}\times J_{1}\times\cdots\times J_{M}},$
$\mathcal{E}_i\in\mathbb{H}^{I_{1}\times\cdots\times I_{N}\times J_{1}\times\cdots\times J_{N}},i=\overline{1,4}$.

In this artical  our purpose is to present a treatment of methods that are useful in the study of solving
Sylvester tensor systems
\begin{equation}\label{mainsystem01}
\begin{cases}
\mathcal{A}_1\ast_{N}\mathcal{X}_1+
\mathcal{Y}_1\ast_{M}\mathcal{B}_1+\mathcal{C}_1\ast_{N}
\mathcal{Z}_1\ast_{M}\mathcal{D}_1+\mathcal{F}_1\ast_{N}
\mathcal{Z}_2\ast_{M}\mathcal{G}_1=\mathcal{E}_1&\\
\mathcal{A}_2\ast_{N}\mathcal{X}_2+
\mathcal{Y}_2\ast_{M}\mathcal{B}_2+\mathcal{C}_2\ast_{N}
\mathcal{Z}_2\ast_{M}\mathcal{D}_2+\mathcal{F}_2\ast_{N}
\mathcal{Z}_3\ast_{M}\mathcal{G}_2=\mathcal{E}_2&\\
\mathcal{A}_3\ast_{N}\mathcal{X}_3+
\mathcal{Y}_3\ast_{M}\mathcal{B}_3+\mathcal{C}_3\ast_{N}
\mathcal{Z}_3\ast_{M}\mathcal{D}_3+\mathcal{F}_3\ast_{N}
\mathcal{Z}_4\ast_{M}\mathcal{G}_3=\mathcal{E}_3,
\end{cases}
\end{equation}
where \begin{equation*}
\begin{split}&\mathcal{A}_{i}\in\mathbb{H}^{I_{1}\times\cdots\times I_{N}\times K_{1}\times\cdots\times K_{N}},~\mathcal{B}_{i}\in\mathbb{H}^{O_{1}\times\cdots\times O_{M}\times J_{1}\times\cdots\times J_{M}},\\&\mathcal{C}_{i}\in\mathbb{H}^{I_{1}\times\cdots\times I_{N}\times Q_{1}\times\cdots\times Q_{N}},~\mathcal{D}_{i}\in\mathbb{H}^{P_{1}\times\cdots\times P_{M}\times J_{1}\times\cdots\times J_{M}},
\\&\mathcal{F}_{i}\in\mathbb{H}^{I_{1}\times\cdots\times I_{N}\times L_{1}\times\cdots\times L_{N}},~\mathcal{G}_{i}\in\mathbb{H}^{S_{1}\times\cdots\times S_{M}\times J_{1}\times\cdots\times J_{M}},\\
&\mathcal{E}_{i}\in\mathbb{H}^{I_{1}\times\cdots\times I_{N}\times J_{1}\times\cdots\times J_{N}},~i=1,2,3.\end{split}
\end{equation*}
we establish the necessary and sufficient conditions for the existence of general solution in terms of  Moore-Penrose inverse, and present  general solution to System (\ref{mainsystem01}).

During investigation about (\ref{mainsystem01}), we are motivatied to find that it is useful for solving  the following Sylvester quaternion tensor system with five variables:
\begin{equation}\label{mainsystem02}
\begin{cases}
\mathcal{A}_1\ast_{N}\mathcal{Z}_1\ast_{M}\mathcal{B}_1+
\mathcal{C}_1\ast_{N}\mathcal{Z}_{2}\ast_{M}\mathcal{D}_1=
\mathcal{E}_1\\
\mathcal{A}_2\ast_{N}\mathcal{Z}_2\ast_{M}\mathcal{B}_2+
\mathcal{C}_2\ast_{N}\mathcal{Z}_{3}\ast_{M}\mathcal{D}_2=
\mathcal{E}_2\\
\mathcal{A}_3\ast_{N}\mathcal{Z}_3\ast_{M}\mathcal{B}_3+
\mathcal{C}_3\ast_{N}\mathcal{Z}_{4}\ast_{M}\mathcal{D}_3=
\mathcal{E}_3\\
\mathcal{A}_4\ast_{N}\mathcal{Z}_4\ast_{M}\mathcal{B}_4+
\mathcal{C}_4\ast_{N}\mathcal{Z}_{5}\ast_{M}\mathcal{D}_4=
\mathcal{E}_4,
\end{cases}
\end{equation}
where $\mathcal{A}_i\in\mathbb{H}^{I_{1}\times\cdots\times I_{N}\times Q_{1}\times\cdots\times Q_{N}},$ $\mathcal{B}_i\in\mathbb{H}^{P_{1}\times\cdots\times P_{M}\times J_{1}\times\cdots\times J_{M}},$
$\mathcal{C}_i\in\mathbb{H}^{I_{1}\times\cdots\times I_{N}\times L_{1}\times\cdots\times L_{N}},$ $\mathcal{D}_i\in\mathbb{H}^{S_{1}\times\cdots\times S_{M}\times J_{1}\times\cdots\times J_{M}},$
$\mathcal{E}_i\in\mathbb{H}^{I_{1}\times\cdots\times I_{N}\times J_{1}\times\cdots\times J_{N}}$.
Also, if $\mathcal{B}_i,\mathcal{C}_i$ are not all zeros, the  system (\ref{mainsystem02}) can also serve as an extended form to the system (\ref{mainsystem00}).

The  paper is divided, structually, into three parts. The first part   (Section 2) contain  some basic notations, definitions and lemmas as tools. The second part   (Section 3) give the solvability conditions and general solutions to quaternion tensor system (\ref{mainsystem01}).  In the third part (Section 4), we provide the solvability conditions and  general solutions to quaternion tensor system (\ref{mainsystem02}).

\section{Preliminaries}
A tensor $\mathcal{A}=(a_{i_{1}\ldots i_{N}})_{1\leq i_{j}\leq I_{j}}$
$(j=1,\ldots,N)$ of order $N$ is a multidimensional array with $I_{1}I_{2}\cdots I_{N}$ entries,
where $N$ is a positive integer. The sets of tensors of order $N$ with dimension $I_{1}\times I_{2}\times\cdots
\times I_{N}$ over the complex field $\mathbb{C}$, the real field $\mathbb{R}$ and the real quaternion algebra
\[\mathbb{H}=\{q_{0}+q_{1}\mathbf{i}+q_{2}\mathbf{j}+q_{3}\mathbf{k}~|~
\mathbf{i}^{2}=\mathbf{j}^{2}=\mathbf{k}^{2}=ijk=-1,q_{0},q_{1},q_{2},q_{3}\in \mathbb{R}\}\]
are represented, respectively, by $\mathbb{C}^{I_{1}\times I_{2}\times\cdots
\times I_{N}}$, $\mathbb{R}^{I_{1}\times I_{2}\times\cdots\times I_{N}}$ and $\mathbb{H}^{I_{1}\times I_{2}\times\cdots\times I_{N}}$. There are more definitions and propositions of quaternions refer to the book \cite{lr} and the paper \cite{fzz}.

A tensor can be viewed as a matrix if $N=2$. For review and convenience in reference, this section contains a summary of necessary facts about
 Moore-Penrose inverse, $\{i\}$-inverse
 of quaternion tensors over $\mathbb{H}^{I_{1}\times\cdots\times I_{N}\times K_{1}
 \times\cdots\times K_{N}}$  as well as some basic definitions related the Einstein product.
Beyond this, we survey some special quaternion tensor systems and their solutions when they are solvable.

\begin{definition}\cite{ea}
Let $\mathcal{A}\in\mathbb{H}^{I_{1}\times\cdots \times I_{P}\times K_{1}
\times\cdots\times K_{N}},\mathcal{B}\in\mathbb{H}^{K_{1}\times\cdots\times
K_{N}\times J_{1}\times\cdots\times J_{M}}$, their Einstein~product is satisfying the associative law, which is
\[(\mathcal{A}\ast_{N}\mathcal{B})_{i_{1}\ldots i_{P}j_{1}\ldots j_{M}}=
\sum_{k_{1},\ldots, k_{N}=1}^{K_{1},\ldots, K_{N}}a_{i_{1}\ldots i_{P}k_{1}
\ldots k_{N}}b_{k_{1}\ldots k_{N}j_{1}\ldots j_{M}},\]
where $\mathcal{A}\ast_{N}\mathcal
{B}\in\mathbb{H}^{I_{1}\times\cdots \times I_{P}\times J_{1}\times\cdots\times
J_{M}}$.
\end{definition}
For simplicity we will denote the above summation by $\sum_{k_{1}\ldots k_{N}}$ or $\sum\limits_{k_{1}\ldots k_{N}}$.  When $N=P=M=1$, we have that $\mathcal{A},\mathcal{B}$ are quaternion matrices, and their Einstein~product is the usual matrix product.

\begin{definition}\cite{zhh}
Let $\mathcal{A}\in\mathbb{H}^{I_{1}\times\cdots\times I_{N}\times K_{1}\times\ldots
\times K_{N}}$, the tensor $\mathcal{X}\in \mathbb{H}^{K_{1}\times\cdots\times K_{N}\times I_{1}\times\cdots\times I_{N}}$ satisfying
\begin{align*}
&\mathrm{(1)}~ \mathcal{A}\ast_{N}\mathcal{X}\ast_{N}\mathcal{A}=\mathcal{A};&\\
&\mathrm{(2)}~
\mathcal{X}\ast_{N}\mathcal{A}\ast_{N}\mathcal{X}=\mathcal{X};&\\
&\mathrm{(3)}~(\mathcal{A}\ast_{N}\mathcal{X})^{*}=\mathcal{A}\ast_{N}\mathcal{X};&\\
&\mathrm{(4)}~(\mathcal{X}\ast_{N}\mathcal{A})^{*}=\mathcal{X}\ast_{N}\mathcal{A},
\end{align*}
is called the Moore-Penrose inverse of $\mathcal{A}$, abbreviated by $\mathrm{M}$-$\mathrm{P}$ inverse, denoted by $\mathcal{A}^{\dag}$.
\end{definition}
The M-P inverse of $\mathcal{A}$ exists and
 is unique. If $(i)$, $i=1, 2, 3, 4$ of the above
 equation holds, then $\mathcal{X}$ is called an $\{i\}$-inverse of $\mathcal{A}$,
  denoted by $\mathcal{A}^{(i)}$.
Furthermore, $\mathcal{L_{A}}$ and $\mathcal{R_{A}}$ stand for the two projectors
 $\mathcal{L_{A}}=\mathcal{I}-\mathcal{A}^{\dag}\ast_{N}\mathcal{A}$ and
 $\mathcal{R_{A}}=\mathcal{I}-\mathcal{A}\ast_{N}\mathcal{A}^{\dag}$ induced
 by $\mathcal{A}$, respectively.
  We say the tensor $\mathcal{B}\in\mathbb{H}^{I_{1}\times\cdots\times I_{N}\times
  I_{1}\times\cdots\times I_{N}}$ is the
inverse of tensor $\mathcal{A}\in\mathbb{H}^{I_{1}\times\cdots\times I_{N}\times
  I_{1}\times\cdots\times I_{N}}$, if $\mathcal{A}\ast_{N}\mathcal{B}=\mathcal{I}=
\mathcal{B}\ast_{N}\mathcal{A}$, and we denote $\mathcal{B}=\mathcal{A}^{-1}$ \cite{bm}.
Moreover, we will say that a tensor is nonsingular if it has an inverse. For an invertible
 tensor $\mathcal{A}$, $\mathcal{A}^{\dag}=\mathcal{A}^{(i)}=\mathcal{A}^{-1}$.

Given a tensor \[\mathcal{A}=(a_{i_{1}\ldots i_{N}j_{1}\ldots j_{M}})\in
\mathbb{H}^{I_{1}\times\cdots \times I_{N}\times J_{1}\times\cdots\times J_{M}},\]
 the tensor $\mathcal{B}=(b_{i_{1}\ldots i_{M}j_{1}\ldots j_{N}})\in\mathbb{H}^{J_{1}
\times\cdots\times J_{M}\times I_{1}\times\cdots\times I_{N}}$ is called
the conjugate transpose of $\mathcal{A}$, and it is denoted by $\mathcal{A}^{*}$, where $b_{i_{1}\ldots i_{M}j_{1}\ldots j_{N}}=\overline{a}_{j_{1}\ldots j_{N}i_{1}\ldots i_{M}}$. The quaternion $\overline{q}=q_{0}-q_{1}\mathbf{i}-q_{2}\mathbf{j}-q_{3}\mathbf{k}$ stands for the conjugate of quaternion $q=q_{0}+q_{1}\mathbf{i}+q_{2}\mathbf{j}+q_{3}\mathbf{k}$, $q_{0},q_{1},q_{2},q_{3}\in \mathbb{R}$.
The tensor $\mathcal{B}=(a_{j_{1}\ldots j_{N}i_{1}\ldots i_{M}})\in\mathbb{H}^{J_{1}\times\cdots\times J_{M}\times I_{1}\times\cdots\times I_{N}}$ is called the transpose of $\mathcal{A}$, and it is denoted by $\mathcal{A}^{T}$. We say that $\mathcal{D}\in\mathbb{H}^{I_{1}\times\cdots\times I_{N}\times J_{1}\times\cdots\times J_{N}}$ is a diagonal tensor if $z\neq i_{1}\ldots i_{N}i_{1}\ldots i_{N}$, then $d_{z}=0$.  $\mathcal{D}$ is a unit tensor, if it is diagonal and $d_{i_{1}\ldots i_{N}i_{1}\ldots i_{N}}=1$, it is denoted by $\mathcal{I}$. We define the trace of a tensor $\mathcal{A}=(a_{i_{1}\ldots i_{N}j_{1}\ldots j_{N}})\in
\mathbb{H}^{I_{1}\times\cdots\times I_{N}\times I_{1}\times\cdots\times I_{N}}$ by $\mathrm{tr}(\mathcal{A})=\Sigma_{i_{1}\ldots i_{N}}
a_{i_{1}\ldots i_{N}i_{1}\ldots i_{N}}$ \cite{bm}. Moreover, we say that matricization is the process of transforming a tensor into a matrix, that is a reordering of the elements of an order $N$ tensor into a matrix, this is also called unfolding or flattening. For example, a  $2\times3\times4\times8$ tensor can be matricized into a $12\times16$ matrix or a $6\times32$ matrix, and so on \cite{tgk}.

The following results can be verified easily.
\begin{proposition}\cite{slz}
Let $\mathcal{A}\in\mathbb{H}^{I_{1}\times\cdots I_{P}\times K_{1}\times\cdots\times K_{N}}$ and $\mathcal{B}\in\mathbb{H}^{K_{1}\times\cdots K_{N}\times J_{1}\times\cdots\times J_{M}}$. Then

$\mathrm{(1)}$ $(\mathcal{A}\ast_{N}\mathcal{B})^{*}=\mathcal{B}^{*}\ast_{N}\mathcal{A}^{*};$

$\mathrm{(2)}$
$\mathcal{I}_{N}\ast_{N}\mathcal{B}=\mathcal{B}$ and $\mathcal{B}\ast_{M}\mathcal{I}_{M}=\mathcal{B}$, where unit tensors $\mathcal{I}_{N}\in\mathbb{H}^{K_{1}\times\cdots\times K_{N}\times K_{1}\times\cdots\times K_{N}}$ and $\mathcal{I}_{M}\in\mathbb{H}^{J_{1}\times\cdots\times J_{M}\times J_{1}\times\cdots\times J_{M}}$.
\end{proposition}
\begin{proposition}
Let $\mathcal{A}\in\mathbb{H}^{I_{1}\times\cdots\times I_{N}\times K_{1}\times\cdots\times K_{N}}$. Then

$\mathrm{(1)}$ $(\mathcal{A}^{\dag})^{\dag}=\mathcal{A}$;

$\mathrm{(2)}$ $(\mathcal{A}^{\dag})^{*}=(A^{*})^{\dag}$;

$\mathrm{(3)}$ $(\mathcal{A}^{*}\ast_{N}\mathcal{A})^{\dag}=\mathcal{A}^{\dag}\ast_{N}(\mathcal{A}^{*})^{\dag}$, $(\mathcal{A}\ast_{N}\mathcal{A}^{*})^{\dag}=(\mathcal{A}^{*})^{\dag}\ast_{N}\mathcal{A}^{\dag}$;

$\mathrm{(4)}$ $\mathcal{A}^{\dag}\ast_{N}\mathcal{R_{A}}=0$ and $\mathcal{R_{A}}\ast_{N}\mathcal{A}=0$.

\end{proposition}
\begin{proposition}\cite{slz}
The following equalities hold:

$\mathrm{(1)}$ $\begin{bmatrix}
                  \mathcal{A}_{1} & \mathcal{B}_{1} \\
                  \mathcal{A}_{2} & \mathcal{B}_{2} \\
                \end{bmatrix}\ast_{M}\begin{bmatrix}
                                       \mathcal{C} \\
                                       \mathcal{D} \\
                                     \end{bmatrix}=\begin{bmatrix}
                                                     \mathcal{A}_{1}\ast_{M}\mathcal{C}+\mathcal{B}_{1}\ast_{M}\mathcal{D} \\
                                               \mathcal{A}_{2}\ast_{M}\mathcal{C}+\mathcal{B}_{2}\ast_{M}\mathcal{D}        \\
                                                   \end{bmatrix}\in\mathbb{H}^{\rho_{1}\times\cdots\times\rho_{N}\times\alpha^{N}}
$;

$\mathrm{(2)}$ $\begin{bmatrix}
                  \mathcal{G} & \mathcal{H} \\
                \end{bmatrix}\ast_{N}\begin{bmatrix}
                  \mathcal{A}_{1} & \mathcal{B}_{1} \\
                  \mathcal{A}_{2} & \mathcal{B}_{2} \\
                \end{bmatrix}=\begin{bmatrix}
                                \mathcal{G}\ast_{N}\mathcal{A}_{1}+\mathcal{H}\ast_{N}\mathcal{A}_{2} & \mathcal{G}\ast_{N}\mathcal{B}_{1}+\mathcal{H}\ast_{N}\mathcal{B}_{2}  \\
                              \end{bmatrix}\in\mathbb{H}^{S_{1}\times\cdots\times S_{N}\times\beta_{1}\times\cdots\times\beta_{M}}
$,
where $\mathcal{A}_{1}\in\mathbb{H}^{I_{1}\times\cdots\times I_{N}\times J_{1}\times\cdots\times J_{M}},$ $\mathcal{B}_{1}\in\mathbb{H}^{I_{1}\times\cdots\times I_{N}\times K_{1}\times\cdots\times K_{M}},$  $\mathcal{A}_{2}\in\mathbb{H}^{L_{1}\times\cdots\times L_{N}\times J_{1}\times\cdots\times J_{M}},$ $\mathcal{B}_{2}\in\mathbb{H}^{L_{1}\times\cdots\times L_{N}\times K_{1}\times\cdots\times K_{M}}$, $\mathcal{C}\in\mathbb{H}^{J_{1}\times\cdots\times J_{M}\times I_{1}\times\cdots\times I_{N}}, \mathcal{D}\in\mathbb{H}^{K_{1}\times\cdots\times K_{M}\times I_{1}\times\cdots\times I_{N}},  \mathcal{G}\in\mathbb{H}^{S_{1}\times\cdots\times S_{N}\times I_{1}\times\cdots\times I_{N}},$ $  \mathcal{H}\in\mathbb{H}^{S_{1}\times\cdots\times S_{N}\times L_{1}\times\cdots\times L_{N}}, \alpha^{N}=I_{1}\times\cdots\times I_{N},$ and $\rho_{i}=I_{i}+L_{i},i=1,\ldots,N; \beta_{j}=J_{j}+K_{j},j=1,\ldots,M$.
\end{proposition}
\begin{lemma}\cite{zhh}\label{lemma0}
Consider tensor equation $\mathcal{A}\ast_{N}\mathcal{X}\ast_{M}\mathcal{B}=C$. Let
$\mathcal{A}\in\mathbb{H}^{I_{1}\times\cdots\times I_{N}\times J_{1}\times\cdots\times J_{N}},~\mathcal{B}\in\mathbb{H}^{K_{1}\times\cdots\times K_{M}\times L_{1}\times\cdots\times L_{M}},$ $\mathcal{C}\in\mathbb{H}^{I_{1}\times\cdots\times I_{N}\times L_{1}\times\cdots\times L_{M}}.$ Then the quaternion tensor equation is consistent if and only if
\[\mathcal{R}_{\mathcal{A}}\ast_{N}\mathcal{C}=0,~
\mathcal{C}\ast_{M}\mathcal{L}_{\mathcal{B}}=0.\]
In this case, the general solution can be expressed as
\[\mathcal{X}=\mathcal{A}^{\dag}\ast_{N}\mathcal{C}
\ast_{M}\mathcal{B}^{\dag}
+\mathcal{L}_{\mathcal{A}}\ast_{N}\mathcal{U}
+\mathcal{V}\ast_{M}\mathcal{R}_{\mathcal{B}},\]
where $\mathcal{U}$ and $\mathcal{V}$ are arbitrary quaternion tensors with appropriate order.
\end{lemma}
\begin{lemma}\cite{zhh}\label{lemma1}
Let $\mathcal{A}\in \mathbb{H}^{I_{1}\times\cdots\times I_{N}\times J_{1}\times\cdots\times J_{N}},~\mathcal{B}\in \mathbb{H}^{K_{1}\times\cdots\times K_{M}\times L_{1}\times\cdots\times L_{M}},~\mathcal{C}\in \mathbb{H}^{I_{1}\times\cdots\times I_{N}\times J_{1}\times\cdots\times J_{N}},$ $\mathcal{D}\in\mathbb{H}^{H_{1}\times\cdots\times H_{M}\times L_{1}\times\cdots\times L_{M}}$ and $\mathcal{E}\in \mathbb{H}^{I_{1}\times\cdots\times I_{N}\times L_{1}\times\cdots\times L_{M}}$. Set $\mathcal{P}=(\mathcal{R}_{\mathcal{A}})\ast_{N}\mathcal{C},
~\mathcal{Q}=\mathcal{D}\ast_{M}(\mathcal{L}_{\mathcal{B}}),
~\mathcal{S}=\mathcal{C}\ast_{N}(\mathcal{L}_{\mathcal{P}})$. Then the generalized Sylvester quaternion tensor equation
\begin{equation}\label{e3}
\mathcal{A}\ast_{N}\mathcal{X}\ast_{M}\mathcal{B}+\mathcal{C}\ast_{N}\mathcal{Y}\ast_{M}\mathcal{D}=\mathcal{E}
\end{equation}
is consistent if and only if
\[\mathcal{R}_{\mathcal{P}}\ast_{N}\mathcal{R}_{\mathcal{A}}\ast_{N}
\mathcal{E}=0,\mathcal{E}\ast_{M}\mathcal{L}_{\mathcal{B}}\ast_{M}
\mathcal{L}_{\mathcal{Q}}=0,\mathcal{R_{A}}\ast_{N}\mathcal{E}
\ast_{M}\mathcal{L_{D}}=0,\mathcal{R_{C}}\ast_{N}\mathcal{E}
\ast_{M}\mathcal{L_{B}}=0\]
In this case, the general solution can be expressed as
\begin{equation*}
\begin{split}
\mathcal{X}&=\mathcal{A}^{\dagger}\ast_{N}\mathcal{E}\ast_{M}
\mathcal{B}^{\dagger}-\mathcal{A}^{\dagger}\ast_{N}\mathcal{C}\ast_{N}
\mathcal{P}^{\dagger}\ast_{N}\mathcal{E}\ast_{M}\mathcal{B}^{\dagger}
-\mathcal{A}^{\dagger}\ast_{N}\mathcal{S}\ast_{N}\mathcal{C}^{\dagger}
\ast_{N}\mathcal{E}\ast_{M}\mathcal{Q}^{\dagger}\ast_{M}\mathcal{D}
\ast_{M}\mathcal{B}^{\dagger}\\
&-\mathcal{A}^{\dagger}\ast_{N}\mathcal{S}\ast_{N}
\mathcal{U}_2\ast_{M}\mathcal{R}_{\mathcal{Q}}\ast_{M}\mathcal{D}
\ast_{M}\mathcal{B}^{\dagger}+\mathcal{L_{A}}\ast_{N}\mathcal{U}_{4}
+\mathcal{U}_{5}\ast_{M}\mathcal{R}_{\mathcal{B}},\\
\mathcal{Y}&=\mathcal{P}^{\dagger}\ast_{N}\mathcal{E}\ast_{M}
\mathcal{D}^{\dagger}
+\mathcal{S}^{\dagger}\ast_{N}\mathcal{S}\ast_{N}\mathcal{C}^{\dagger}
\ast_{N}\mathcal{E}\ast_{M}\mathcal{Q}^{\dagger}
+\mathcal{L}_{\mathcal{P}}\ast_{N}\mathcal{L}_{\mathcal{S}}\ast_{N}\mathcal{U}_{1}
+\mathcal{L}_{\mathcal{P}}\ast_{N}\mathcal{U}_{2}+\mathcal{U}_{3}\ast_{M}\mathcal{R_{D}},
\end{split}
\end{equation*}
where $\mathcal{U}_{i},~i=1,\ldots,5$ are arbitrary quaternion tensors with appropriate order.
\end{lemma}

\section{Solvable conditions and general solution to the  system (\ref{mainsystem01})}
In this section we consider the solvability conditions and the expression of the general solution to the quaternion tensor system (\ref{mainsystem01}).  For simplicity, put
\begin{align}
\mathcal{A}_{ii}&=\mathcal{R}_{\mathcal{A}_{i}}
\ast_{N}\mathcal{C}_{i},
~\mathcal{B}_{ii}=\mathcal{D}_{i}\ast_{M}
\mathcal{L}_{\mathcal{B}_{i}},
~\mathcal{C}_{ii}=\mathcal{R}_{\mathcal{A}_{i}}\ast_{N}
\mathcal{F}_{i},
~\mathcal{D}_{ii}=\mathcal{G}_{i}\ast_{M}
\mathcal{L}_{\mathcal{B}_{i}},\label{july17equ001}\\
\mathcal{E}_{ii}&=\mathcal{R}_{\mathcal{A}_{i}}
\ast_{N}\mathcal{E}_{i}\ast_{M}\mathcal{L}_{\mathcal{B}_{i}},
~\mathcal{M}_{ii}=\mathcal{R}_{\mathcal{A}_{ii}}
\ast_{N}\mathcal{C}_{ii},\label{july17equ002}\\
~\mathcal{N}_{ii}&=\mathcal{D}_{ii}\ast_{M}
\mathcal{L}_{\mathcal{B}_{ii}},
~\mathcal{S}_{ii}=\mathcal{C}_{ii}\ast_{N}
\mathcal{L}_{\mathcal{M}_{ii}},
~i=1,~2,~3,\label{july17equ003}\\
\mathcal{A}_{j,j+1}&=\begin{bmatrix}
          \mathcal{L}_{\mathcal{M}_{jj}}\ast_{N}
          \mathcal{L}_{\mathcal{S}_{jj}} & -\mathcal{L}_{\mathcal{A}_{j+1,j+1}} \\
        \end{bmatrix},~\mathcal{B}_{j,j+1}=\begin{bmatrix}
                                \mathcal{R}_{\mathcal{D}_{jj}} \\
                                -\mathcal{R}_{\mathcal{B}_{j+1,j+1}} \\
                              \end{bmatrix},\label{july17equ004}\\
~\mathcal{C}_{j,j+1}&=\mathcal{L}_{\mathcal{M}_{jj}},
~\mathcal{D}_{j,j+1}=\mathcal{R}_{\mathcal{N}_{jj}},\label{july17equ005}\\
\mathcal{F}_{j,j+1}&=\mathcal{A}_{j+1,j+1}^{\dagger}
\ast_{N}\mathcal{S}_{j+1,j+1},
~\mathcal{G}_{j,j+1}=\mathcal{R}_{\mathcal{N}_{j+1,j+1}}\ast_{M}
\mathcal{D}_{j+1,j+1}\ast_{M}\mathcal{B}_{j+1,j+1}^{\dagger},\label{july17equ006}\\
\mathcal{E}_{j,j+1}&=\mathcal{M}_{jj}^{\dagger}\ast_{N}\mathcal{E}_{jj}
\ast_{M}\mathcal{D}_{jj}^{\dagger}
+\mathcal{S}_{jj}^{\dagger}\ast_{N}\mathcal{S}_{jj}\ast_{N}
\mathcal{C}_{jj}^{\dagger}\ast_{N}
\mathcal{E}_{jj}\ast_{M}\mathcal{N}_{jj}^{\dagger}-
\mathcal{A}_{j+1,j+1}^{\dagger}\ast_{N}\mathcal{E}_{j+1,j+1}\nonumber\\
&\ast_{M}
\mathcal{B}_{j+1,j+1}^{\dagger}
+\mathcal{A}_{j+1,j+1}^{\dagger}\ast_{N}
\mathcal{C}_{j+1,j+1}
\ast_{N}\mathcal{M}_{j+1,j+1}^{\dagger}\ast_{N}\mathcal{E}_{j+1,j+1}\ast_{M}
\mathcal{B}_{j+1,j+1}^{\dagger}
+\mathcal{A}_{j+1,j+1}^{\dagger}\nonumber\\
&\ast_{N}\mathcal{S}_{j+1,j+1}\ast_{N}
\mathcal{C}_{j+1,j+1}^{\dagger}\ast_{N}\mathcal{E}_{j+1,j+1}\ast_{M}
\mathcal{N}_{j+1,j+1}^{\dagger}\ast_{M}\mathcal{D}_{j+1,j+1}\ast_{M}
\mathcal{B}_{j+1,j+1}^{\dagger},\label{july17equ007}\\
\widehat{\mathcal{A}_{j,j+1}}&=R_{\mathcal{A}_{j,j+1}}\ast_{N}
\mathcal{C}_{j,j+1},
~\widehat{\mathcal{B}_{j,j+1}}=\mathcal{D}_{j,j+1}\ast_{M}
\mathcal{L}_{\mathcal{B}_{j,j+1}},\label{july17equ008}\\
~\widehat{\mathcal{C}_{j,j+1}}&=\mathcal{R}_{\mathcal{A}_{j,j+1}}
\ast_{N}\mathcal{F}_{j,j+1},
~\widehat{\mathcal{D}_{j,j+1}}=\mathcal{G}_{j,j+1}\ast_{M}
\mathcal{L}_{\mathcal{B}_{j,j+1}},\label{july17equ009}\\
\widehat{\mathcal{E}_{j,j+1}}&=\mathcal{R}_{\mathcal{A}_{j,j+1}}
\ast_{N}\mathcal{E}_{j,j+1}\ast_{M}\mathcal{L}_{\mathcal{B}_{j,j+1}},
~\widehat{\mathcal{M}_{j,j+1}}=\mathcal{R}_{\widehat{\mathcal{A}_{j,j+1}}}
\ast_{N}\widehat{\mathcal{C}_{j,j+1}},\label{july17equ010}\\
~\widehat{\mathcal{N}_{j,j+1}}&=\widehat{\mathcal{D}_{j,j+1}}\ast_{M}
\mathcal{L}_{\widehat{\mathcal{B}_{j,j+1}}},
~\widehat{\mathcal{S}_{j,j+1}}=\widehat{\mathcal{C}_{j,j+1}}\ast_{N}
\mathcal{L}_{\widehat{\mathcal{M}_{j,j+1}}},~j=1,~2,\label{july17equ011}
\\
\mathcal{A}&=\begin{bmatrix}
          \mathcal{L}_{\widehat{\mathcal{M}_{12}}}\ast_{N}
          \mathcal{L}_{\widehat{\mathcal{S}_{12}}} & -\mathcal{L}_{\widehat{\mathcal{A}_{23}}} \\
        \end{bmatrix},~\mathcal{B}=\begin{bmatrix}
                                \mathcal{R}_{\widehat{\mathcal{D}_{12}}} \\
                                -\mathcal{R}_{\widehat{\mathcal{B}_{23}}} \\
                              \end{bmatrix},\label{july17equ012}\\
~\mathcal{C}&=\mathcal{L}_{\widehat{\mathcal{M}_{12}}},
~\mathcal{D}=\mathcal{R}_{\widehat{\mathcal{N}_{12}}},&
\\
\mathcal{F}&=\widehat{\mathcal{A}_{23}}^{\dagger}\ast_{N}
\widehat{\mathcal{S}_{23}},
~\mathcal{G}=\mathcal{R}_{\widehat{\mathcal{N}_{23}}}
\ast_{M}\widehat{\mathcal{D}_{23}}
\ast_{M}\widehat{\mathcal{B}_{23}}^{\dagger},&
\end{align}
\begin{align}
\mathcal{E}&=\widehat{\mathcal{M}_{12}}^{\dagger}\ast_{N}
\widehat{\mathcal{E}_{12}}\ast_{M}\widehat{\mathcal{D}_{12}}^{\dagger}+
\widehat{\mathcal{S}_{12}}^{\dagger}\ast_{N}\widehat{\mathcal{S}_{12}}\ast_{N}
\widehat{\mathcal{C}_{12}}^{\dagger}\ast_{N}
\widehat{\mathcal{E}_{12}}\ast_{M}\widehat{\mathcal{N}_{12}}^{\dagger}-
\widehat{\mathcal{A}_{23}}^{\dagger}\ast_{N}\widehat{\mathcal{E}_{23}}
\nonumber\\
&\ast_{M}\widehat{\mathcal{B}_{23}}^{\dagger}+
\widehat{\mathcal{A}_{23}}^{\dagger}\ast_{N}\widehat{\mathcal{C}_{23}}\ast_{N}
\widehat{\mathcal{M}_{23}}^{\dagger}\ast_{N}
\widehat{\mathcal{E}_{23}}\ast_{M}\widehat{\mathcal{B}_{23}}^{\dagger}
+\widehat{\mathcal{A}_{23}}^{\dagger}\ast_{N}\widehat{\mathcal{S}_{23}}\ast_{N}
\widehat{\mathcal{C}_{23}}^{\dagger}\ast_{N}
\widehat{\mathcal{E}_{23}}\nonumber\\
&\ast_{M}\widehat{\mathcal{N}_{23}}^{\dagger}\ast_{M}
\widehat{\mathcal{D}_{23}}\ast_{M}
\widehat{\mathcal{B}_{23}}^{\dagger},\label{july17equ015}\\
\widehat{\mathcal{A}}&=\mathcal{R}_{\mathcal{A}}\ast_{N}\mathcal{C},
~\widehat{\mathcal{B}}=\mathcal{D}\ast_{M}\mathcal{L}_{\mathcal{B}},
~\widehat{\mathcal{C}}=\mathcal{R}_{\mathcal{A}}\ast_{N}\mathcal{F},
~\widehat{\mathcal{D}}=\mathcal{G}_{12}\ast_{M}\mathcal{L}_{\mathcal{B}},\label{july17equ016}\\
\widehat{\mathcal{E}}&=\mathcal{R}_{\mathcal{A}}\ast_{N}
\mathcal{E}\ast_{M}\mathcal{L}_{\mathcal{B}},
~\widehat{\mathcal{M}}=\mathcal{R}_{\widehat{\mathcal{A}}}
\ast_{N}\widehat{\mathcal{C}},
~\widehat{\mathcal{N}}=\widehat{\mathcal{D}}\ast_{M}
\mathcal{L}_{\widehat{\mathcal{B}}},
~\widehat{\mathcal{S}}=\widehat{\mathcal{C}}\ast_{N}
\mathcal{L}_{\widehat{\mathcal{M}}},\label{july17equ017}
\end{align}
Before giving the basic theorems, we present a lemma which is helpful to prove the Theorem \ref{thm1}.
\begin{lemma}\label{lemma4}
Let $\mathcal{A}_{11}=\mathcal{R}_{\mathcal{A}_{1}}
\ast_{N}\mathcal{C}_{1},
~\mathcal{B}_{11}=\mathcal{D}_{1}\ast_{M}
\mathcal{L}_{\mathcal{B}_{1}},
~\mathcal{C}_{11}=\mathcal{R}_{\mathcal{A}_{1}}\ast_{N}
\mathcal{F}_{1},
~\mathcal{D}_{11}=\mathcal{G}_{1}\ast_{M}
\mathcal{L}_{\mathcal{B}_{1}},
\mathcal{E}_{11}=\mathcal{R}_{\mathcal{A}_{1}}
\ast_{N}\mathcal{E}_{1}\ast_{M}\mathcal{L}_{\mathcal{B}_{1}},
~\mathcal{M}_{11}=\mathcal{R}_{\mathcal{A}_{11}}
\ast_{N}\mathcal{C}_{11},
~\mathcal{N}_{11}=\mathcal{D}_{11}\ast_{M}
\mathcal{L}_{\mathcal{B}_{11}},
~\mathcal{S}_{11}=\mathcal{C}_{11}\ast_{N}
\mathcal{L}_{\mathcal{M}_{11}}.$
Consider the following tensor equation
\begin{equation}\label{july23equ01}
\mathcal{A}_1\ast_{N}\mathcal{X}_1+
\mathcal{Y}_1\ast_{M}\mathcal{B}_1+\mathcal{C}_1\ast_{N}
\mathcal{Z}_1\ast_{M}\mathcal{D}_1+\mathcal{F}_1\ast_{N}
\mathcal{Z}_2\ast_{M}\mathcal{G}_1=\mathcal{E}_1,
\end{equation}
which is consistent if and only if
\begin{eqnarray}
&\mathcal{R}_{\mathcal{M}_{11}}\ast_{N}
\mathcal{R}_{\mathcal{A}_{11}}\ast_{N}\mathcal{E}_{11}=0,
~\mathcal{E}_{11}\ast_{M}\mathcal{L}_{\mathcal{B}_{11}}
\ast_{M}\mathcal{L}_{\mathcal{N}_{11}}=0,\label{july23equ02}\\
&\mathcal{R}_{\mathcal{A}_{11}}\ast_{N}\mathcal{E}_{11}
\ast_{M}\mathcal{L}_{\mathcal{D}_{11}}=0,
~\mathcal{R}_{\mathcal{C}_{11}}\ast_{N}
\mathcal{E}_{11}\ast_{M}\mathcal{L}_{\mathcal{B}_{11}}=0,\label{july23equ03}
\end{eqnarray}
In this case, the general solution can be expressed as
\begin{align}
\mathcal{X}_{1}&=\mathcal{A}_{1}^{\dag}\ast_{N}(\mathcal{E}_{1}-\mathcal{C}_{1}
\ast_{N}\mathcal{Z}_{1}\ast_{M}
\mathcal{D}_{1}-\mathcal{F}_{1}\ast_{N}\mathcal{Z}_{2}\ast_{M}
\mathcal{G}_{1})-\mathcal{T}_{1}\ast_{M}\mathcal{B}_{1}
+\mathcal{L}_{\mathcal{A}_{1}}\ast_{N}\mathcal{T}_{2},\label{july23equ04}\\
\mathcal{Y}_{1}&=\mathcal{R}_{\mathcal{A}_{1}}\ast_{N}
(\mathcal{E}_{1}-\mathcal{C}_{1}\ast_{N}\mathcal{Z}_{1}
\ast_{M}\mathcal{D}_{1}-
\mathcal{F}_{1}\ast_{N}\mathcal{Z}_{2}\ast_{M}\mathcal{G}_{1})
\ast_{M}\mathcal{B}_{1}^{\dag}
+\mathcal{A}_{1}\ast_{N}\mathcal{T}_{1}+
\mathcal{T}_{3}
\ast_{M}\mathcal{R}_{\mathcal{B}_{1}},\label{july23equ05}\\
\mathcal{Z}_1&=\mathcal{A}_{11}^{\dagger}\ast_{N}\mathcal{E}_{11}
\ast_{M}\mathcal{B}_{11}^{\dagger}-\mathcal{A}_{11}^{\dagger}
\ast_{N}\mathcal{C}_{11}\ast_{N}
\mathcal{M}_{11}^{\dagger}\ast_{N}\mathcal{E}_{11}\ast_{M}
\mathcal{B}_{11}^{\dagger}
-\mathcal{A}_{11}^{\dagger}\ast_{N}\mathcal{S}_{11}\ast_{N}
\mathcal{C}_{11}^{\dagger}
\ast_{N}\mathcal{E}_{11}\ast_{M}\mathcal{N}_{11}^{\dagger}
\ast_{M}\nonumber\\
&\mathcal{D}_{11}\ast_{M}\mathcal{B}_{11}^{\dagger}-
\mathcal{A}_{11}^{\dagger}
\ast_{N}\mathcal{S}_{11}\ast_{N}
\mathcal{T}_{4}\ast_{M}\mathcal{R}_{\mathcal{N}_{11}}\ast_{M}\mathcal{D}_{ii}
\ast_{M}\mathcal{B}_{11}^{\dagger}+
\mathcal{L}_{\mathcal{A}_{11}}\ast_{N}\mathcal{T}_{5}
+\mathcal{T}_{6}\ast_{M}\mathcal{R}_{\mathcal{B}_{11}},\label{july23equ06}
\\
\mathcal{Z}_{2}&=\mathcal{M}_{11}^{\dagger}\ast_{N}
\mathcal{E}_{11}\ast_{M}
\mathcal{D}_{11}^{\dagger}
+\mathcal{S}_{11}^{\dagger}\ast_{N}\mathcal{S}_{11}\ast_{N}
\mathcal{C}_{11}^{\dagger}
\ast_{N}\mathcal{E}_{11}\ast_{M}\mathcal{N}_{11}^{\dagger}
+\mathcal{L}_{\mathcal{M}_{11}}\ast_{N}\mathcal{L}_{\mathcal{S}_{11}}
\ast_{N}\mathcal{T}_{7}
+\mathcal{L}_{\mathcal{M}_{11}}\ast_{N}\mathcal{T}_{4}\nonumber\\
&\ast_{M}\mathcal{R}_{\mathcal{N}_{11}}+
\mathcal{T}_{8}\ast_{M}\mathcal{R}_{\mathcal{D}_{11}},\label{july23equ07}
\end{align}
where $T_{i}~(i=1,\ldots,8)$ are arbitrary quaternion tensors over $\mathbb{H}$.
\end{lemma}

\begin{proof}
We seperate the left part of equation (\ref{july23equ01}) into two parts $\mathcal{A}_1\ast_{N}\mathcal{X}_1+
\mathcal{Y}_1\ast_{M}\mathcal{B}_{1}$ and $\mathcal{C}_1\ast_{N}
\mathcal{Z}_1\ast_{M}\mathcal{D}_1+\mathcal{F}_1\ast_{N}
\mathcal{Z}_2\ast_{M}\mathcal{G}_1$, we have
\begin{equation}\label{july23equ08}
\mathcal{A}_1\ast_{N}\mathcal{X}_1+\mathcal{Y}_1\ast_{M}\mathcal{B}_1=
\mathcal{E}_1-\mathcal{C}_1\ast_{N}
\mathcal{Z}_1\ast_{M}\mathcal{D}_1-\mathcal{F}_1\ast_{N}
\mathcal{Z}_2\ast_{M}\mathcal{G}_1.
\end{equation}
Applying Lemma \ref{lemma1} if $\mathcal{B}=\mathcal{D}=0$, equation (\ref{july23equ08}) is consistent if and only if 
\begin{equation}\label{july23equ09}
\mathcal{R}_{\mathcal{A}_{1}}\ast_{N}(\mathcal{E}_1-\mathcal{C}_1\ast_{N}
\mathcal{Z}_1\ast_{M}\mathcal{D}_1-\mathcal{F}_1\ast_{N}
\mathcal{Z}_2\ast_{M}\mathcal{G}_1)\ast_{M}\mathcal{L}_{\mathcal{B}_{1}}=0,
\end{equation} 
$\mathcal{X}_1,\mathcal{Y}_1$ can be written as (\ref{july23equ04}), (\ref{july23equ05}).
That is equalivent to prove
\begin{equation}\label{july23equ10}
\mathcal{A}_{11}\ast_{N}
\mathcal{Z}_1\ast_{M}\mathcal{B}_{11}+\mathcal{C}_{11}\ast_{N}
\mathcal{Z}_2\ast_{M}\mathcal{D}_{11}=\mathcal{E}_{11}
\end{equation}
is consistent. Applying Lemma \ref{lemma1},  that equation (\ref{july23equ10}) is consistent if and only if (\ref{july23equ02}) and (\ref{july23equ03}),  as well as $\mathcal{Z}_1,\mathcal{Z}_2$ can be written  as (\ref{july23equ06}), (\ref{july23equ07}).
	\end{proof}
 When the $N=M=1$, the matrix form of above-mentioned Lemma was given in \cite{zhh0}.
\begin{theorem}\label{thm1}
Consider system (\ref{mainsystem01}).
Then the following statements are equivalent:
\begin{enumerate}
	\item\label{item01} System (\ref{mainsystem01}) is consistent.
	\item\label{item02} The following equalities are satisfied:
	\begin{align}
	&\label{july17equ018}\mathcal{R}_{\mathcal{M}_{ii}}\ast_{N}
	\mathcal{R}_{\mathcal{A}_{ii}}\ast_{N}\mathcal{E}_{ii}=0,
	~\mathcal{E}_{ii}\ast_{M}\mathcal{L}_{\mathcal{B}_{ii}}
	\ast_{M}\mathcal{L}_{\mathcal{N}_{i}}=0,\\
	&\label{july17equ019}\mathcal{R}_{\mathcal{A}_{ii}}\ast_{N}\mathcal{E}_{ii}
	\ast_{M}\mathcal{L}_{\mathcal{D}_{ii}}=0,
	~\mathcal{R}_{\mathcal{C}_{ii}}\ast_{N}
	\mathcal{E}_{ii}\ast_{M}\mathcal{L}_{\mathcal{B}_{ii}}=0, ~(i=1,2,3),\\
	\label{july17equ020}&\mathcal{R}_{\widehat{\mathcal{M}_{j,j+1}}}\ast_{N}
	\mathcal{R}_{\widehat{\mathcal{A}_{j,j+1}}}\ast_{N}
	\widehat{\mathcal{E}_{j,j+1}}=0,
	~\widehat{\mathcal{E}_{j,j+1}}\ast_{M}
	\mathcal{L}_{\widehat{\mathcal{B}_{j,j+1}}}
	\ast_{M}\mathcal{L}_{\widehat{\mathcal{Q}_{j,j+1}}}=0,\\
	\label{july17equ021}&\mathcal{R}_{\widehat{\mathcal{A}_{j,j+1}}}\ast_{N}
	\widehat{\mathcal{E}_{j,j+1}}
	\ast_{M}\mathcal{L}_{\widehat{\mathcal{B}_{j,j+1}}}=0,
	~\mathcal{R}_{\widehat{\mathcal{C}_{j,j+1}}}\ast_{N}
	\widehat{\mathcal{E}_{j,j+1}}\ast_{M}\mathcal{L}_{\widehat{\mathcal{B}_{j,j+1}}}=0, ~(j=1,2),\\
	\label{july17equ022}&\mathcal{R}_{\widehat{\mathcal{M}}}\ast_{N}
	\mathcal{R}_{\widehat{\mathcal{A}}}\ast_{N}\widehat{\mathcal{E}}=0,
	~\widehat{\mathcal{E}}\ast_{M}\mathcal{L}_{\widehat{\mathcal{B}}}
	\ast_{M}\mathcal{L}_{\widehat{\mathcal{Q}}}=0,\\
	\label{july17equ023}&\mathcal{R}_{\widehat{\mathcal{A}}}\ast_{N}
	\widehat{\mathcal{E}}
	\ast_{M}\mathcal{L}_{\widehat{\mathcal{B}}}=0,
	~\mathcal{R}_{\widehat{\mathcal{C}}}\ast_{N}
	\widehat{\mathcal{E}}\ast_{M}\mathcal{L}_{\widehat{\mathcal{B}}}=0,
	\end{align}
\end{enumerate}
Furthermore, if statement (\ref{item01}) holds, then the general solution to (\ref{mainsystem01}) can be expressed as
\begin{align}
\label{july17equ024}
\mathcal{X}_{i}&=\mathcal{A}_{i}^{\dag}\ast_{N}(\mathcal{E}_{i}-\mathcal{C}_{i}
\ast_{N}\mathcal{Z}_{i}\ast_{M}
\mathcal{D}_{i}-\mathcal{F}_{i}\ast_{N}\mathcal{Z}_{i+1}\ast_{M}
\mathcal{G}_{i})-\mathcal{T}_{i1}\ast_{M}\mathcal{B}_{i}
+\mathcal{L}_{\mathcal{A}_{i}}\ast_{N}\mathcal{T}_{i2},&\\
\label{july17equ025}
\mathcal{Y}_{i}&=\mathcal{R}_{\mathcal{A}_{i}}\ast_{N}
(\mathcal{E}_{i}-\mathcal{C}_{i}\ast_{N}\mathcal{Z}_{i}
\ast_{M}\mathcal{D}_{i}-
\mathcal{F}_{i}\ast_{N}\mathcal{Z}_{i+1}\ast_{M}\mathcal{G}_{i})
\ast_{M}\mathcal{B}_{i}^{\dag}
+\mathcal{A}_{i}\ast_{N}\mathcal{T}_{i1}+
\mathcal{T}_{i3}
\ast_{M}\mathcal{R}_{\mathcal{B}_{i}},&\\
\label{july17equ026}
\mathcal{Z}_i&=\mathcal{A}_{ii}^{\dagger}\ast_{N}\mathcal{E}_{ii}\ast_{M}
\mathcal{B}_{ii}^{\dagger}-\mathcal{A}_{ii}^{\dagger}\ast_{N}\mathcal{C}_{ii}\ast_{N}
\mathcal{M}_{ii}^{\dagger}\ast_{N}\mathcal{E}_{ii}\ast_{M}\mathcal{B}_{ii}^{\dagger}
-\mathcal{A}_{ii}^{\dagger}\ast_{N}\mathcal{S}_{ii}\ast_{N}\mathcal{C}_{ii}^{\dagger}
\ast_{N}\mathcal{E}_{ii}\ast_{M}\mathcal{N}_{ii}^{\dagger}\ast_{M}\nonumber\\
&\mathcal{D}_{ii}\ast_{M}\mathcal{B}_{ii}^{\dagger}-\mathcal{A}_{ii}^{\dagger}
\ast_{N}\mathcal{S}_{ii}\ast_{N}
\mathcal{T}_{i4}\ast_{M}\mathcal{R}_{\mathcal{N}_{ii}}\ast_{M}\mathcal{D}_{ii}
\ast_{M}\mathcal{B}_{ii}^{\dagger}+
\mathcal{L}_{\mathcal{A}_{ii}}\ast_{N}\mathcal{T}_{i5}
+\mathcal{T}_{i6}\ast_{M}\mathcal{R}_{\mathcal{B}_{ii}},&\\
\label{july17equ027}
\mathcal{Z}_{i+1}&=\mathcal{M}_{ii}^{\dagger}\ast_{N}
\mathcal{E}_{ii}\ast_{M}
\mathcal{D}_{ii}^{\dagger}
+\mathcal{S}_{ii}^{\dagger}\ast_{N}\mathcal{S}_{ii}\ast_{N}
\mathcal{C}_{ii}^{\dagger}
\ast_{N}\mathcal{E}_{ii}\ast_{M}\mathcal{N}_{ii}^{\dagger}
+\mathcal{L}_{\mathcal{M}_{ii}}\ast_{N}\mathcal{L}_{\mathcal{S}_{ii}}
\ast_{N}\mathcal{T}_{i7}
+\mathcal{L}_{\mathcal{M}_{ii}}\ast_{N}\mathcal{T}_{i4}\nonumber\\
&\ast_{M}\mathcal{R}_{\mathcal{N}_{ii}}+
\mathcal{T}_{i8}\ast_{M}\mathcal{R}_{\mathcal{D}_{ii}},~(i=1,2,3)
\end{align}
where
\begin{align}
\label{july17equ028}
\mathcal{T}_{j7}&=\begin{bmatrix}
                    \mathcal{I}_{m_j} & 0 \\
                  \end{bmatrix}\ast_{N}
[\mathcal{A}_{j,j+1}^{\dag}\ast_{N}(\mathcal{E}_{j,j+1}-\mathcal{C}_{j,j+1}
\ast_{N}\mathcal{T}_{j4}\ast_{M}
\mathcal{D}_{j,j+1}-\mathcal{F}_{j,j+1}\ast_{N}\mathcal{T}_{j+1,4}\ast_{M}
\mathcal{G}_{j,j+1})\nonumber\\
&-\mathcal{U}_{j1}\ast_{M}\mathcal{B}_{j,j+1}
+\mathcal{L}_{\mathcal{A}_{j,j+1}}\ast_{N}\mathcal{U}_{j2}],&\\
\label{july17equ029}
\mathcal{T}_{j+1,5}&=\begin{bmatrix}
                    0 & \mathcal{I}_{m_j} \\
                  \end{bmatrix}\ast_{N}
[\mathcal{A}_{j,j+1}^{\dag}\ast_{N}(\mathcal{E}_{j,j+1}-\mathcal{C}_{j,j+1}
\ast_{N}\mathcal{T}_{j4}\ast_{M}
\mathcal{D}_{j,j+1}-\mathcal{F}_{j,j+1}\ast_{N}\mathcal{T}_{j+1,4}\ast_{M}
\mathcal{G}_{j,j+1})\nonumber\\
&-\mathcal{U}_{j1}\ast_{M}\mathcal{B}_{j,j+1}
+\mathcal{L}_{\mathcal{A}_{j,j+1}}\ast_{N}\mathcal{U}_{j2}],&\\
\label{july17equ030}
\mathcal{T}_{j8}&=[\mathcal{R}_{\mathcal{A}_{j,j+1}}\ast_{N}
(\mathcal{E}_{j,j+1}-\mathcal{C}_{j,j+1}\ast_{N}\mathcal{U}_{12}
\ast_{M}\mathcal{D}_{j,j+1}-
\mathcal{F}_{j,j+1}\ast_{N}\mathcal{U}_{22}\ast_{M}\mathcal{G}_{j,j+1})
\ast_{M}\mathcal{B}_{j,j+1}^{\dag}\nonumber\\
&+\mathcal{A}_{j,j+1}\ast_{N}\mathcal{U}_{j1}+
\mathcal{U}_{j3}
\ast_{M}\mathcal{R}_{\mathcal{B}_{j,j+1}}]\ast_{M}
\begin{bmatrix}
\mathcal{I}_{n_j} \\
0 \\
\end{bmatrix},&\\
\label{july17equ031}
\mathcal{T}_{j+1,6}&=[\mathcal{R}_{\mathcal{A}_{j,j+1}}\ast_{N}
(\mathcal{E}_{j,j+1}-\mathcal{C}_{j,j+1}\ast_{N}\mathcal{U}_{12}
\ast_{M}\mathcal{D}_{j,j+1}-
\mathcal{F}_{j,j+1}\ast_{N}\mathcal{U}_{22}\ast_{M}\mathcal{G}_{j,j+1})
\ast_{M}\mathcal{B}_{j,j+1}^{\dag}\nonumber\\
&+\mathcal{A}_{j,j+1}\ast_{N}\mathcal{U}_{j1}+
\mathcal{U}_{j3}
\ast_{M}\mathcal{R}_{\mathcal{B}_{j,j+1}}]\ast_{M}
\begin{bmatrix}
0\\
\mathcal{I}_{n_j}  \\
\end{bmatrix},&
\end{align}
\begin{align}
\label{july17equ032}
\mathcal{T}_{j4}&=\widehat{\mathcal{A}_{j,j+1}}^{\dagger}\ast_{N}
\widehat{\mathcal{E}_{j,j+1}}\ast_{M}\widehat{\mathcal{B}_{j,j+1}}^{\dagger}
-\widehat{\mathcal{A}_{j,j+1}}^{\dagger}\ast_{N}\widehat{\mathcal{C}_{j,j+1}}\ast_{N}
\widehat{\mathcal{M}_{j,j+1}}^{\dagger}\ast_{N}\widehat{\mathcal{E}_{j,j+1}}
\ast_{M}\mathcal{B}_{j,j+1}^{\dagger}-\nonumber\\
&\widehat{\mathcal{A}_{j,j+1}}^{\dagger}\ast_{N}\widehat{\mathcal{S}_{j,j+1}}
\ast_{N}\widehat{\mathcal{C}_{j,j+1}}^{\dagger}
\ast_{N}\widehat{\mathcal{E}_{j,j+1}}\ast_{M}
\widehat{\mathcal{N}_{j,j+1}}^{\dagger}\ast_{M}
\widehat{\mathcal{D}_{j,j+1}}\ast_{M}
\widehat{\mathcal{B}_{j,j+1}}^{\dagger}
-\widehat{\mathcal{A}_{j,j+1}}^{\dagger}\ast_{N}\nonumber\\
&\widehat{\mathcal{S}_{j,j+1}}\ast_{N}
\widehat{\mathcal{U}_{j4}}\ast_{M}\mathcal{R}_{\widehat{\mathcal{N}_{j,j+1}}}\ast_{M}
\widehat{\mathcal{D}_{j,j+1}}
\ast_{M}\widehat{\mathcal{B}_{j,j+1}}^{\dagger}
+\mathcal{L}_{\widehat{\mathcal{A}_{j,j+1}}}\ast_{N}\mathcal{U}_{j5}
+\mathcal{U}_{j6}\ast_{M}\mathcal{R}_{\widehat{\mathcal{B}_{j,j+1}}},&\\
\label{july17equ033}
\mathcal{T}_{j+1,4}&=\widehat{\mathcal{M}_{j,j+1}}^{\dag}\ast_{N}
\widehat{\mathcal{E}_{j,j+1}}\ast_{N}\widehat{\mathcal{D}_{j,j+1}}^{\dag}
+\widehat{\mathcal{S}_{j,j+1}}^{\dag}\ast_{N}\widehat{\mathcal{S}_{j,j+1}}
\ast_{N}\widehat{\mathcal{C}_{j,j+1}}\ast_{N}
\widehat{\mathcal{E}_{j,j+1}}\ast_{M}
\widehat{\mathcal{N}_{j,j+1}}^{\dag}+\nonumber\\
&\mathcal{L}_{\widehat{\mathcal{M}}_{j,j+1}}\ast_{N}
\mathcal{L}_{\widehat{\mathcal{S}}_{j,j+1}}
\ast_{N}\mathcal{U}_{j7}
+\mathcal{L}_{\widehat{\mathcal{M}_{j,j+1}}}\ast_{N}
\mathcal{U}_{j4}\ast_{M}\mathcal{R}_{\widehat{\mathcal{N}_{j,j+1}}}+
\mathcal{U}_{j8}\ast_{M}\mathcal{R}_{\widehat{\mathcal{D}_{j,j+1}}},
~(j=1,2)&\\
\label{july17equ034}
\mathcal{U}_{17}&=\begin{bmatrix}
                    \mathcal{I}_{s} & 0 \\
                  \end{bmatrix}\ast_{N}
[\mathcal{A}^{\dag}\ast_{N}(\mathcal{E}-\mathcal{C}
\ast_{N}\mathcal{U}_{14}\ast_{M}
\mathcal{D}-\mathcal{F}\ast_{N}\mathcal{U}_{24}\ast_{M}
\mathcal{G})
-\mathcal{V}_{1}\ast_{M}\mathcal{B}
+\mathcal{L}_{\mathcal{A}}\ast_{N}\mathcal{V}_{2}],&\\
\label{july17equ035}
\mathcal{U}_{25}&=\begin{bmatrix}
                    0 & \mathcal{I}_{s} \\
                  \end{bmatrix}\ast_{N}
[\mathcal{A}^{\dag}\ast_{N}(\mathcal{E}-\mathcal{C}
\ast_{N}\mathcal{U}_{14}\ast_{M}
\mathcal{D}-\mathcal{F}\ast_{N}\mathcal{U}_{24}\ast_{M}
\mathcal{G})
-\mathcal{V}_{1}\ast_{M}\mathcal{B}
+\mathcal{L}_{\mathcal{A}}\ast_{N}\mathcal{V}_{3}],&\\
\label{july17equ036}
\mathcal{U}_{18}&=[\mathcal{R}_{\mathcal{A}}\ast_{N}
(\mathcal{E}-\mathcal{C}\ast_{N}\mathcal{U}_{14}
\ast_{M}\mathcal{D}-
\mathcal{F}\ast_{N}\mathcal{U}_{24}\ast_{M}\mathcal{G})
\ast_{M}\mathcal{B}^{\dag}
+\mathcal{A}\ast_{N}\mathcal{U}_{j1}+
\mathcal{U}_{j3}
\ast_{M}\mathcal{R}_{\mathcal{B}}]\nonumber\\
&\ast_{M}
\begin{bmatrix}
\mathcal{I}_{t} \\
0 \\
\end{bmatrix},&\\
\label{july17equ037}
\mathcal{U}_{26}&=[\mathcal{R}_{\mathcal{A}}\ast_{N}
(\mathcal{E}-\mathcal{C}\ast_{N}\mathcal{U}_{14}
\ast_{M}\mathcal{D}-
\mathcal{F}\ast_{N}\mathcal{U}_{24}\ast_{M}\mathcal{G})
\ast_{M}\mathcal{B}^{\dag}
+\mathcal{A}\ast_{N}\mathcal{U}_{j1}+
\mathcal{U}_{j3}
\ast_{M}\mathcal{R}_{\mathcal{B}}]\nonumber\\
&\ast_{M}
\begin{bmatrix}
0\\
\mathcal{I}_{t}  \\
\end{bmatrix},&\\
\label{july17equ038}
\mathcal{U}_{14}&=\widehat{\mathcal{A}}^{\dagger}\ast_{N}
\widehat{\mathcal{E}}\ast_{M}\widehat{\mathcal{B}}^{\dagger}
-\widehat{\mathcal{A}}^{\dagger}\ast_{N}\widehat{\mathcal{C}}\ast_{N}
\widehat{\mathcal{M}}^{\dagger}\ast_{N}\widehat{\mathcal{E}}
\ast_{M}\mathcal{B}^{\dagger}-
\widehat{\mathcal{A}}^{\dagger}\ast_{N}\widehat{\mathcal{S}}
\ast_{N}\widehat{\mathcal{C}}^{\dagger}
\ast_{N}\widehat{\mathcal{E}}\ast_{M}
\widehat{\mathcal{N}}^{\dagger}\ast_{M}
\widehat{\mathcal{D}}\ast_{M}
\widehat{\mathcal{B}}^{\dagger}\nonumber\\
&-\widehat{\mathcal{A}}^{\dagger}\ast_{N}
\widehat{\mathcal{S}}\ast_{N}
\mathcal{V}_{4}\ast_{M}
\mathcal{R}_{\widehat{\mathcal{N}}}\ast_{M}
\widehat{\mathcal{D}}
\ast_{M}\widehat{\mathcal{B}}^{\dagger}
+\mathcal{L}_{\widehat{\mathcal{A}}}\ast_{N}\mathcal{V}_{5}
+\mathcal{V}_{6}\ast_{M}\mathcal{R}_{\widehat{\mathcal{B}}},&\\
\label{july17equ039}
\mathcal{U}_{24}&=\widehat{\mathcal{M}}^{\dag}\ast_{N}
\widehat{\mathcal{E}}\ast_{N}\widehat{\mathcal{D}}^{\dag}
+\widehat{\mathcal{S}}^{\dag}\ast_{N}\widehat{\mathcal{S}}
\ast_{N}\widehat{\mathcal{C}}\ast_{N}
\widehat{\mathcal{E}}\ast_{M}
\widehat{\mathcal{N}}^{\dag}+
\mathcal{L}_{\widehat{\mathcal{M}}}\ast_{N}
\mathcal{L}_{\widehat{\mathcal{S}}}
\ast_{N}\mathcal{V}_{7}
+\mathcal{L}_{\widehat{\mathcal{M}}}\ast_{N}
\mathcal{V}_{4}\ast_{M}
\mathcal{R}_{\widehat{\mathcal{N}}}\nonumber\\
&+\mathcal{V}_{8}\ast_{M}\mathcal{R}_{\widehat{\mathcal{D}}},
\end{align}
and $\mathcal{T}_{i1},~\mathcal{T}_{i2},~\mathcal{T}_{i3},
~\mathcal{U}_{j1},~\mathcal{U}_{j2},~\mathcal{U}_{j3},~
~\mathcal{T}_{14},~\mathcal{T}_{15},~
\mathcal{T}_{16},~\mathcal{T}_{37},~\mathcal{T}_{38},~ \mathcal{U}_{15},~\mathcal{U}_{16},~~\mathcal{U}_{27},$
$~\mathcal{U}_{28},
~\mathcal{V}_{t},$ are arbitrary tensors with appropriate sizes over $\mathbb{H}$, $m_j,s$ is the same as the column block of $\mathcal{F}_j,\mathcal{F}_2$, respectively,  $n_j,t$ is the same as the row block of $\mathcal{G}_j,\mathcal{G}_2$, respectively (i=1,2,3; j=1,2; ~t=1,\ldots,8).
\end{theorem}
\begin{proof}
(\ref{item01}) $\Longleftrightarrow$ (\ref{item02}).

We separate the tensor system into three groups
\begin{align}
\label{july17equ040}\mathcal{A}_1\ast_{N}\mathcal{X}_1+
\mathcal{Y}_1\ast_{M}\mathcal{B}_1+\mathcal{C}_1\ast_{N}
\mathcal{Z}_1\ast_{M}\mathcal{D}_1+\mathcal{F}_1\ast_{N}
\mathcal{Z}_2\ast_{M}\mathcal{G}_1=\mathcal{E}_1,\\
\label{july17equ041}\mathcal{A}_2\ast_{N}\mathcal{X}_2+
\mathcal{Y}_2\ast_{M}\mathcal{B}_2+\mathcal{C}_2\ast_{N}
\mathcal{Z}_2\ast_{M}\mathcal{D}_2+\mathcal{F}_2\ast_{N}
\mathcal{Z}_3\ast_{M}\mathcal{G}_2=\mathcal{E}_2,
\end{align}
and
\begin{align}
\label{july17equ042}\mathcal{A}_3\ast_{N}\mathcal{X}_3+
\mathcal{Y}_3\ast_{M}\mathcal{B}_3+\mathcal{C}_3\ast_{N}
\mathcal{Z}_3\ast_{M}\mathcal{D}_3+\mathcal{F}_3\ast_{N}
\mathcal{Z}_4\ast_{M}\mathcal{G}_3=\mathcal{E}_3.
\end{align}

It follows from the Lemma \ref{lemma4} that (\ref{july17equ040}), (\ref{july17equ041}) and (\ref{july17equ042}) are consistent, respectively, if and only if (\ref{july17equ018}) and (\ref{july17equ019}) hold. The general solution of (\ref{july17equ040}), (\ref{july17equ041}), (\ref{july17equ042}) can be expressed as (\ref{july17equ024}), (\ref{july17equ025}), (\ref{july17equ026}), (\ref{july17equ027}),
where $\mathcal{T}_{i1}\cdots\mathcal{T}_{i8},~i=1,2,3$ are arbitrary appropriate size  tensor over $\mathbb{H}$.

Let $\mathcal{Z}_2$ in (\ref{july17equ027}) (when $i=1$) be equal to $\mathcal{Z}_2$ in (\ref{july17equ026}) (when $i=2$) and $\mathcal{Z}_3$ in (\ref{july17equ027}) (when $i=2$) be equal to $\mathcal{Z}_3$ in (\ref{july17equ026}) (when $i=3$).
Note
\begin{itemize}
  \item When $i=1$,
\begin{align}\label{july17equ043}
\mathcal{Z}_{i+1}(\mathcal{Z}_{2})&=\mathcal{M}_{11}^{\dagger}\ast_{N}
\mathcal{E}_{11}\ast_{M}
\mathcal{D}_{11}^{\dagger}
+\mathcal{S}_{11}^{\dagger}\ast_{N}\mathcal{S}_{11}\ast_{N}
\mathcal{C}_{11}^{\dagger}
\ast_{N}\mathcal{E}_{11}\ast_{M}\mathcal{N}_{11}^{\dagger}
+\mathcal{L}_{\mathcal{M}_{11}}\ast_{N}\mathcal{L}_{\mathcal{S}_{11}}
\ast_{N}\mathcal{T}_{27}\nonumber\\
&+\mathcal{L}_{\mathcal{M}_{11}}\ast_{N}\mathcal{T}_{24}
\ast_{M}\mathcal{R}_{\mathcal{N}_{11}}+
\mathcal{T}_{28}\ast_{M}\mathcal{R}_{\mathcal{D}_{11}},&
\end{align}
\item When $i=2$,
\begin{align}\label{july17equ044}
\mathcal{Z}_{i}(\mathcal{Z}_2)&=\mathcal{A}_{22}^{\dagger}\ast_{N}\mathcal{E}_{22}\ast_{M}
\mathcal{B}_{22}^{\dagger}-\mathcal{A}_{22}^{\dagger}\ast_{N}
\mathcal{C}_{22}\ast_{N}
\mathcal{M}_{22}^{\dagger}\ast_{N}\mathcal{E}_{22}\ast_{M}
\mathcal{B}_{22}^{\dagger}
-\mathcal{A}_{22}^{\dagger}\ast_{N}\mathcal{S}_{22}\ast_{N}
\mathcal{C}_{22}^{\dagger}
\ast_{N}\mathcal{E}_{22}\nonumber\\
&\ast_{M}\mathcal{N}_{22}^{\dagger}
\ast_{M}
\mathcal{D}_{22}\ast_{M}\mathcal{B}_{22}^{\dagger}-\mathcal{A}_{22}^{\dagger}
\ast_{N}\mathcal{S}_{22}\ast_{N}
\mathcal{T}_{24}\ast_{M}\mathcal{R}_{\mathcal{N}_{22}}\ast_{M}\mathcal{D}_{22}
\ast_{M}\mathcal{B}_{22}^{\dagger}+
\mathcal{L}_{\mathcal{A}_{22}}\ast_{N}\mathcal{T}_{25}\nonumber\\
&+\mathcal{T}_{26}\ast_{M}\mathcal{R}_{\mathcal{B}_{22}},
\end{align}
\begin{align}\label{july17equ045}
\mathcal{Z}_{i+1}(\mathcal{Z}_{3})&=\mathcal{M}_{22}^{\dagger}\ast_{N}
\mathcal{E}_{22}\ast_{M}
\mathcal{D}_{22}^{\dagger}
+\mathcal{S}_{22}^{\dagger}\ast_{N}\mathcal{S}_{22}\ast_{N}
\mathcal{C}_{22}^{\dagger}
\ast_{N}\mathcal{E}_{22}\ast_{M}\mathcal{N}_{22}^{\dagger}
+\mathcal{L}_{\mathcal{M}_{22}}\ast_{N}\mathcal{L}_{\mathcal{S}_{22}}
\ast_{N}\mathcal{T}_{37}\nonumber\\
&+\mathcal{L}_{\mathcal{M}_{22}}\ast_{N}\mathcal{T}_{34}
\ast_{M}\mathcal{R}_{\mathcal{N}_{22}}+
\mathcal{T}_{38}\ast_{M}\mathcal{R}_{\mathcal{D}_{22}},
\end{align}
\item When $i=3$,
\begin{align}\label{july17equ046}
\mathcal{Z}_i(\mathcal{Z}_3)&=\mathcal{A}_{33}^{\dagger}\ast_{N}\mathcal{E}_{33}\ast_{M}
\mathcal{B}_{33}^{\dagger}-\mathcal{A}_{33}^{\dagger}\ast_{N}
\mathcal{C}_{33}\ast_{N}
\mathcal{M}_{33}^{\dagger}\ast_{N}\mathcal{E}_{33}\ast_{M}
\mathcal{B}_{33}^{\dagger}
-\mathcal{A}_{33}^{\dagger}\ast_{N}\mathcal{S}_{33}\ast_{N}
\mathcal{C}_{33}^{\dagger}
\ast_{N}\mathcal{E}_{33}\nonumber\\
&\ast_{M}\mathcal{N}_{33}^{\dagger}
\ast_{M}
\mathcal{D}_{33}\ast_{M}\mathcal{B}_{33}^{\dagger}-\mathcal{A}_{33}^{\dagger}
\ast_{N}\mathcal{S}_{33}\ast_{N}
\mathcal{T}_{34}\ast_{M}\mathcal{R}_{\mathcal{N}_{33}}\ast_{M}\mathcal{D}_{33}
\ast_{M}\mathcal{B}_{33}^{\dagger}+
\mathcal{L}_{\mathcal{A}_{33}}\ast_{N}\mathcal{T}_{35}\nonumber\\
&+\mathcal{T}_{36}\ast_{M}\mathcal{R}_{\mathcal{B}_{33}}.
\end{align}
\end{itemize}
Then equating (\ref{july17equ043}),  (\ref{july17equ044}) and (\ref{july17equ045}),  (\ref{july17equ046}), respectively. We have the following
\begin{equation}\label{july17equ047}
\begin{split}
\mathcal{A}_{12}\ast_{N}\begin{bmatrix}
\mathcal{T}_{17} \\
\mathcal{T}_{25} \\
\end{bmatrix}+\begin{bmatrix}
\mathcal{T}_{18} & \mathcal{T}_{26} \\
\end{bmatrix}\ast_{M}\mathcal{B}_{12}+\mathcal{C}_{12}\ast_{N}
\mathcal{T}_{14}\ast_{M}\mathcal{D}_{12}+
\mathcal{F}_{12}\ast_{N}\mathcal{T}_{24}\ast_{M}\mathcal{G}_{12}
=\mathcal{E}_{12},
\end{split}
\end{equation}
and
\begin{equation}\label{july17equ048}
\mathcal{A}_{23}\ast_{N}\begin{bmatrix}
\mathcal{T}_{27} \\
\mathcal{T}_{35} \\
\end{bmatrix}+\begin{bmatrix}
\mathcal{T}_{28} & \mathcal{T}_{36} \\
\end{bmatrix}\ast_{M}\mathcal{B}_{23}+\mathcal{C}_{23}\ast_{N}
\mathcal{T}_{24}\ast_{M}\mathcal{D}_{23}+
\mathcal{F}_{23}\ast_{N}\mathcal{T}_{34}\ast_{M}\mathcal{G}_{23}
=\mathcal{E}_{23},
\end{equation}
where $\mathcal{A}_{j,j+1},\mathcal{B}_{j,j+1},\mathcal{C}_{j,j+1},\mathcal{D}_{j,j+1},\mathcal{F}_{j,j+1},\mathcal{E}_{j,j+1}$ are given in $ (\ref{july17equ004}) -(\ref{july17equ007})( j=1,2)$. 

Firstly, we consider the solvability conditions and general solution to the equation (\ref{july17equ047}) and (\ref{july17equ048}),  respectively. 

It follows from Lemma \ref{lemma4} that the equations  (\ref{july17equ047}) and (\ref{july17equ048}) are consistent if and only if (\ref{july17equ020}) -- (\ref{july17equ021}) and the general solution to equations (\ref{july17equ047}) and (\ref{july17equ048}) can be written as (\ref{july17equ028}) -- (\ref{july17equ033}) (j=1,2).

From above two cases, we obtain two expressions of $\mathcal{T}_{24}$  as following
\begin{itemize}
	\item When $j=1$,
	\begin{align}
	\label{july17equ049}
	\mathcal{T}_{j+1,4}(\mathcal{T}_{24})&=\widehat{\mathcal{M}_{12}}^{\dag}\ast_{N}
	\widehat{\mathcal{E}_{12}}\ast_{N}\widehat{\mathcal{D}_{12}}^{\dag}
	+\widehat{\mathcal{S}_{12}}^{\dag}\ast_{N}\widehat{\mathcal{S}_{12}}
	\ast_{N}\widehat{\mathcal{C}_{12}}\ast_{N}
	\widehat{\mathcal{E}_{12}}\ast_{M}
	\widehat{\mathcal{N}_{12}}^{\dag}+
	\mathcal{L}_{\widehat{\mathcal{M}}_{12}}\ast_{N}
	\mathcal{L}_{\widehat{\mathcal{S}}_{12}}
	\ast_{N}\mathcal{U}_{17}\nonumber\\
	&+\mathcal{L}_{\widehat{\mathcal{M}_{12}}}\ast_{N}
	\mathcal{U}_{14}\ast_{M}\mathcal{R}_{\widehat{\mathcal{N}_{12}}}+
	\mathcal{U}_{18}\ast_{M}\mathcal{R}_{\widehat{\mathcal{D}_{12}}},
	\end{align}
	\item When $j=2$,
	\begin{align}
	\label{july17equ050}
	\mathcal{T}_{j,4}(\mathcal{T}_{24})&=\widehat{\mathcal{A}_{23}}^{\dagger}\ast_{N}
	\widehat{\mathcal{E}_{23}}\ast_{M}\widehat{\mathcal{B}_{23}}^{\dagger}
	-\widehat{\mathcal{A}_{23}}^{\dagger}\ast_{N}
	\widehat{\mathcal{C}_{23}}\ast_{N}
	\widehat{\mathcal{M}_{23}}^{\dagger}\ast_{N}
	\widehat{\mathcal{E}_{23}}
	\ast_{M}\mathcal{B}_{23}^{\dagger}-
	\widehat{\mathcal{A}_{23}}^{\dagger}\ast_{N}\widehat{\mathcal{S}_{23}}
	\ast_{N}\widehat{\mathcal{C}_{23}}^{\dagger}\nonumber\\
	&\ast_{N}\widehat{\mathcal{E}_{23}}\ast_{M}
	\widehat{\mathcal{N}_{23}}^{\dagger}\ast_{M}
	\widehat{\mathcal{D}_{23}}\ast_{M}
	\widehat{\mathcal{B}_{23}}^{\dagger}
	-\widehat{\mathcal{A}_{23}}^{\dagger}\ast_{N}
	\widehat{\mathcal{S}_{23}}\ast_{N}
	\widehat{\mathcal{U}_{24}}\ast_{M}\mathcal{R}_{\widehat{\mathcal{N}_{j,j+1}}}\ast_{M}
	\widehat{\mathcal{D}_{23}}
	\ast_{M}\widehat{\mathcal{B}_{23}}^{\dagger}\nonumber\\
	&+\mathcal{L}_{\widehat{\mathcal{A}_{23}}}\ast_{N}\mathcal{U}_{25}
	+\mathcal{U}_{26}\ast_{M}\mathcal{R}_{\widehat{\mathcal{B}_{23}}},
	\end{align}
\end{itemize}

Equating (\ref{july17equ049}) and (\ref{july17equ050}), we have
\begin{equation}\label{july17equ051}
\begin{split}
\mathcal{A}\ast_{N}\begin{bmatrix}
               \mathcal{U}_{17} \\
               \mathcal{U}_{25} \\
             \end{bmatrix}+\begin{bmatrix}
  \mathcal{U}_{18} & \mathcal{U}_{26} \\
\end{bmatrix}\ast_{M}\mathcal{B}+\mathcal{C}\ast_{N}
\mathcal{U}_{14}\ast_{M}\mathcal{D}+
\mathcal{F}\ast_{N}\mathcal{U}_{24}\ast_{M}\mathcal{G}
             =\mathcal{E},
\end{split}
\end{equation}
where $\mathcal{A},\mathcal{B},\mathcal{C},\mathcal{D},\mathcal{F},\mathcal{G},\mathcal{F}$ are given by (\ref{july17equ012}) -- (\ref{july17equ015}). Then we consider the sovability condition and general solution to equation (\ref{july17equ051}).

Using Lemma \ref{lemma4} over again,  equation (\ref{july17equ051}) is consistent if and only if (\ref{july17equ022}) and (\ref{july17equ023}) hold. And $\mathcal{U}_{17},~\mathcal{U}_{25},~\mathcal{U}_{18},
~\mathcal{U}_{26},~\mathcal{U}_{14},~\mathcal{U}_{24}$ can be expressed as (\ref{july17equ034})--(\ref{july17equ039}), where $V_{t},~t=1,\ldots,8$ are arbitrary quaternion tensors.
\end{proof}

\section{Solvable conditions and general solution to the System (\ref{mainsystem02})}
In this section, an general solution to System (\ref{mainsystem02}) using Moore-Penrose is given and the necessary and sufficient conditions are investigated.
For simplicity, put
\begin{align}
\mathcal{M}_k&=\mathcal{R}_{\mathcal{A}_k}\ast_{N}\mathcal{C}_k,
~\mathcal{N}_k=\mathcal{D}_k\ast_{M}\mathcal{L}_{\mathcal{B}_k},
~\mathcal{S}_k=\mathcal{C}_k\ast_{N}\mathcal{L}_{\mathcal{M}_k},
~k=1,\ldots,4,&\\
\widehat{\mathcal{A}_{i}}&=\begin{bmatrix}
          \mathcal{L}_{\mathcal{M}_{i}}\ast_{N}
          \mathcal{L}_{\mathcal{S}_{i}} & -\mathcal{L}_{\mathcal{A}_{i+1}} \\
        \end{bmatrix},~\widehat{\mathcal{B}_{i}}=\begin{bmatrix}
                                \mathcal{R}_{\mathcal{D}_{i}} \\
                                -\mathcal{R}_{\mathcal{B}_{i+1}} \\
                              \end{bmatrix},
~\widehat{\mathcal{C}_{i}}=\mathcal{L}_{\mathcal{M}_{i}},
~\widehat{\mathcal{D}_{i}}=\mathcal{R}_{\mathcal{N}_{i}},&\\
\widehat{\mathcal{F}_{i}}&=\mathcal{A}_{i+1}^{\dagger}
\ast_{N}\mathcal{S}_{i+1},
~\widehat{\mathcal{G}_{i}}=\mathcal{R}_{\mathcal{N}_{i+1}}\ast_{M}
\mathcal{D}_{i+1}\ast_{M}\mathcal{B}_{i+1}^{\dagger},&\\
\mathcal{E}_{i}&=\mathcal{M}_{i}^{\dagger}\ast_{N}\mathcal{E}_{i}
\ast_{M}\mathcal{D}_{i}^{\dagger}
+\mathcal{S}_{i}^{\dagger}\ast_{N}\mathcal{S}_{i}\ast_{N}
\mathcal{C}_{i}^{\dagger}\ast_{N}
\mathcal{E}_{i}\ast_{M}\mathcal{N}_{i}^{\dagger}-
\mathcal{A}_{i+1}^{\dagger}\ast_{N}\mathcal{E}_{i+1}\ast_{M}
\mathcal{B}_{i+1}^{\dagger}\nonumber\\
&+\mathcal{A}_{i+1}^{\dagger}\ast_{N}
\mathcal{C}_{i+1}
\ast_{N}\mathcal{M}_{i+1}^{\dagger}\ast_{N}\mathcal{E}_{i+1}\ast_{M}
\mathcal{B}_{i+1}^{\dagger}\nonumber\\
&+\mathcal{A}_{i+1}^{\dagger}
\ast_{N}\mathcal{S}_{i+1}\ast_{N}
\mathcal{C}_{i+1}^{\dagger}\ast_{N}\mathcal{E}_{i+1}\ast_{M}
\mathcal{N}_{i+1}^{\dagger}\ast_{M}\mathcal{D}_{i+1}\ast_{M}
\mathcal{B}_{i+1}^{\dagger},&\\
\mathcal{A}_{ii}&=\mathcal{R}_{\widehat{\mathcal{A}_{i}}}
\ast_{N}\widehat{\mathcal{C}_{i}},
~\mathcal{B}_{ii}=\widehat{\mathcal{D}_{i}}\ast_{M}
\mathcal{L}_{\widehat{\mathcal{B}_{i}}},
~\mathcal{C}_{ii}=\mathcal{R}_{\widehat{\mathcal{A}_{i}}}\ast_{N}
\widehat{\mathcal{F}_{i}},
~\mathcal{D}_{ii}=\widehat{\mathcal{G}_{i}}\ast_{M}
\mathcal{L}_{\widehat{\mathcal{B}_{i}}},&\\
\mathcal{E}_{ii}&=\mathcal{R}_{\widehat{\mathcal{A}_{i}}}
\ast_{N}\widehat{\mathcal{E}_{i}}\ast_{M}
\mathcal{L}_{\widehat{\mathcal{B}_{i}}},
~\mathcal{M}_{ii}=\mathcal{R}_{\mathcal{A}_{ii}}
\ast_{N}\mathcal{C}_{ii},&\\
~\mathcal{N}_{ii}&=\mathcal{D}_{ii}\ast_{M}
\mathcal{L}_{\mathcal{B}_{ii}},
~\mathcal{S}_{ii}=\mathcal{C}_{ii}\ast_{N}
\mathcal{L}_{\mathcal{M}_{ii}},
~i=1,2,3,
\\
\mathcal{A}_{j,j+1}&=\begin{bmatrix}
          \mathcal{L}_{\mathcal{M}_{jj}}\ast_{N}
          \mathcal{L}_{\mathcal{S}_{jj}} & -\mathcal{L}_{\mathcal{A}_{j+1,j+1}} \\
        \end{bmatrix},~\mathcal{B}_{j,j+1}=\begin{bmatrix}
                                \mathcal{R}_{\mathcal{D}_{jj}} \\
                                -\mathcal{R}_{\mathcal{B}_{j+1,j+1}} \\
                              \end{bmatrix},&\\
~\mathcal{C}_{j,j+1}&=\mathcal{L}_{\mathcal{M}_{jj}},
~\mathcal{D}_{j,j+1}=\mathcal{R}_{\mathcal{N}_{jj}},&\\
\mathcal{F}_{j,j+1}&=\mathcal{A}_{j+1,j+1}^{\dagger}
\ast_{N}\mathcal{S}_{j+1,j+1},
~\mathcal{G}_{j,j+1}=\mathcal{R}_{\mathcal{N}_{j+1,j+1}}\ast_{M}
\mathcal{D}_{j+1,j+1}\ast_{M}\mathcal{B}_{j+1,j+1}^{\dagger},&
\end{align}
\begin{align}
\mathcal{E}_{j,j+1}&=\mathcal{M}_{jj}^{\dagger}\ast_{N}\mathcal{E}_{jj}
\ast_{M}\mathcal{D}_{jj}^{\dagger}
+\mathcal{S}_{jj}^{\dagger}\ast_{N}\mathcal{S}_{jj}\ast_{N}
\mathcal{C}_{jj}^{\dagger}\ast_{N}
\mathcal{E}_{jj}\ast_{M}\mathcal{N}_{jj}^{\dagger}-
\mathcal{A}_{j+1,j+1}^{\dagger}\ast_{N}\mathcal{E}_{j+1,j+1}\nonumber\\
&\ast_{M}
\mathcal{B}_{j+1,j+1}^{\dagger}
+\mathcal{A}_{j+1,j+1}^{\dagger}\ast_{N}
\mathcal{C}_{j+1,j+1}
\ast_{N}\mathcal{M}_{j+1,j+1}^{\dagger}\ast_{N}\mathcal{E}_{j+1,j+1}\ast_{M}
\mathcal{B}_{j+1,j+1}^{\dagger}
+\mathcal{A}_{j+1,j+1}^{\dagger}\nonumber\\
&\ast_{N}\mathcal{S}_{j+1,j+1}\ast_{N}
\mathcal{C}_{j+1,j+1}^{\dagger}\ast_{N}\mathcal{E}_{j+1,j+1}\ast_{M}
\mathcal{N}_{j+1,j+1}^{\dagger}\ast_{M}\mathcal{D}_{j+1,j+1}\ast_{M}
\mathcal{B}_{j+1,j+1}^{\dagger},&\\
\widehat{\mathcal{A}_{j,j+1}}&=R_{\mathcal{A}_{j,j+1}}\ast_{N}
\mathcal{C}_{j,j+1},
~\widehat{\mathcal{B}_{j,j+1}}=\mathcal{D}_{j,j+1}\ast_{M}
\mathcal{L}_{\mathcal{B}_{j,j+1}},&\\
~\widehat{\mathcal{C}_{j,j+1}}&=\mathcal{R}_{\mathcal{A}_{j,j+1}}
\ast_{N}\mathcal{F}_{j,j+1},
~\widehat{\mathcal{D}_{j,j+1}}=\mathcal{G}_{j,j+1}\ast_{M}
\mathcal{L}_{\mathcal{B}_{j,j+1}},&\\
\widehat{\mathcal{E}_{j,j+1}}&=\mathcal{R}_{\mathcal{A}_{j,j+1}}
\ast_{N}\mathcal{E}_{j,j+1}\ast_{M}\mathcal{L}_{\mathcal{B}_{j,j+1}},
~\widehat{\mathcal{M}_{j,j+1}}=\mathcal{R}_{\widehat{\mathcal{A}_{j,j+1}}}
\ast_{N}\widehat{\mathcal{C}_{j,j+1}},&\\
~\widehat{\mathcal{N}_{j,j+1}}&=\widehat{\mathcal{D}_{j,j+1}}\ast_{M}
\mathcal{L}_{\widehat{\mathcal{B}_{j,j+1}}},
~\widehat{\mathcal{S}_{j,j+1}}=\widehat{\mathcal{C}_{j,j+1}}\ast_{N}
\mathcal{L}_{\widehat{\mathcal{M}_{j,j+1}}},~j=1,~2,&\\
\mathcal{A}&=\begin{bmatrix}
          \mathcal{L}_{\widehat{\mathcal{M}_{12}}}\ast_{N}
          \mathcal{L}_{\widehat{\mathcal{S}_{12}}} & -\mathcal{L}_{\widehat{\mathcal{A}_{23}}} \\
        \end{bmatrix},~\mathcal{B}=\begin{bmatrix}
                                \mathcal{R}_{\widehat{\mathcal{D}_{12}}} \\
                                -\mathcal{R}_{\widehat{\mathcal{B}_{23}}} \\
                              \end{bmatrix},&\\
~\mathcal{C}&=\mathcal{L}_{\widehat{\mathcal{M}_{12}}},
~\mathcal{D}=\mathcal{R}_{\widehat{\mathcal{N}_{12}}},&\\
\mathcal{F}&=\widehat{\mathcal{A}_{23}}^{\dagger}\ast_{N}
\widehat{\mathcal{S}_{23}},
~\mathcal{G}=\mathcal{R}_{\widehat{\mathcal{N}_{23}}}
\ast_{M}\widehat{\mathcal{D}_{23}}
\ast_{M}\widehat{\mathcal{B}_{23}}^{\dagger},&\\
\mathcal{E}&=\widehat{\mathcal{M}_{12}}^{\dagger}\ast_{N}
\widehat{\mathcal{E}_{12}}\ast_{M}\widehat{\mathcal{D}_{12}}^{\dagger}+
\widehat{\mathcal{S}_{12}}^{\dagger}\ast_{N}\widehat{\mathcal{S}_{12}}\ast_{N}
\widehat{\mathcal{C}_{12}}^{\dagger}\ast_{N}
\widehat{\mathcal{E}_{12}}\ast_{M}\widehat{\mathcal{N}_{12}}^{\dagger}-
\widehat{\mathcal{A}_{23}}^{\dagger}\ast_{N}\widehat{\mathcal{E}_{23}}
\nonumber\\
&\ast_{M}\widehat{\mathcal{B}_{23}}^{\dagger}+
\widehat{\mathcal{A}_{23}}^{\dagger}\ast_{N}\widehat{\mathcal{C}_{23}}\ast_{N}
\widehat{\mathcal{M}_{23}}^{\dagger}\ast_{N}
\widehat{\mathcal{E}_{23}}\ast_{M}\widehat{\mathcal{B}_{23}}^{\dagger}
+\widehat{\mathcal{A}_{23}}^{\dagger}\ast_{N}\widehat{\mathcal{S}_{23}}\ast_{N}
\widehat{\mathcal{C}_{23}}^{\dagger}\ast_{N}
\widehat{\mathcal{E}_{23}}\nonumber\\
&\ast_{M}\widehat{\mathcal{N}_{23}}^{\dagger}\ast_{M}
\widehat{\mathcal{D}_{23}}\ast_{M}\widehat{\mathcal{B}_{23}}^{\dagger},&\\
\widehat{\mathcal{A}}&=\mathcal{R}_{\mathcal{A}}\ast_{N}\mathcal{C},
~\widehat{\mathcal{B}}=\mathcal{D}\ast_{M}\mathcal{L}_{\mathcal{B}},
~\widehat{\mathcal{C}}=\mathcal{R}_{\mathcal{A}}\ast_{N}\mathcal{F},
~\widehat{\mathcal{D}}=\mathcal{G}_{12}\ast_{M}\mathcal{L}_{\mathcal{B}},&\\
\widehat{\mathcal{E}}&=\mathcal{R}_{\mathcal{A}}\ast_{N}
\mathcal{E}\ast_{M}\mathcal{L}_{\mathcal{B}},
~\widehat{\mathcal{M}}=\mathcal{R}_{\widehat{\mathcal{A}}}
\ast_{N}\widehat{\mathcal{C}},
~\widehat{\mathcal{N}}=\widehat{\mathcal{D}}\ast_{M}
\mathcal{L}_{\widehat{\mathcal{B}}},
~\widehat{\mathcal{S}}=\widehat{\mathcal{C}}\ast_{N}
\mathcal{L}_{\widehat{\mathcal{M}}},
\end{align}
According the proof of Theorem \ref{thm1}, we get the following theorem:
\begin{theorem}\label{thm2}
Consider system (\ref{july17equ022}).
Then the following statements are equivalent:
\begin{enumerate}
	\item\label{item03} System (\ref{july17equ022}) is consistent.
	\item\label{item04} The following equalities are satisfied:
	\begin{align}
	&\label{july17equ052}\mathcal{R}_{\mathcal{M}_{i}}\ast_{N}
	\mathcal{R}_{\mathcal{A}_{i}}\ast_{N}\mathcal{E}_{i}=0,
	~\mathcal{E}_{i}\ast_{M}\mathcal{L}_{\mathcal{B}_{i}}
	\ast_{M}\mathcal{L}_{\mathcal{N}_{i}}=0,\\
	&\label{july17equ053}\mathcal{R}_{\mathcal{A}_{i}}\ast_{N}\mathcal{E}_{i}
	\ast_{M}\mathcal{L}_{\mathcal{D}_{i}}=0,
	~\mathcal{R}_{\mathcal{C}_{i}}\ast_{N}
	\mathcal{E}_{i}\ast_{M}\mathcal{L}_{\mathcal{B}_{i}}=0,&\\
	&\mathcal{R}_{\mathcal{M}_{ii}}\ast_{N}
	\mathcal{R}_{\widehat{\mathcal{A}_{i}}}\ast_{N}
	\widehat{\mathcal{E}_{i}}=0,
	~\widehat{\mathcal{E}_{i}}\ast_{M}
	\mathcal{L}_{\widehat{\mathcal{B}_{i}}}
	\ast_{M}\mathcal{L}_{\mathcal{N}_{ii}}=0,\\
	&\mathcal{R}_{\widehat{\mathcal{A}_{i}}}\ast_{N}\widehat{\mathcal{E}_{i}}
	\ast_{M}\mathcal{L}_{\widehat{\mathcal{D}_{i}}}=0,
	~\mathcal{R}_{\widehat{\mathcal{C}_{i}}}\ast_{N}
	\widehat{\mathcal{E}_{i}}\ast_{M}
	\mathcal{L}_{\widehat{\mathcal{B}_{i}}}=0, ~(i=1,2,3),\\
	&\mathcal{R}_{\widehat{\mathcal{M}_{j,j+1}}}\ast_{N}
	\mathcal{R}_{\widehat{\mathcal{A}_{j,j+1}}}\ast_{N}
	\widehat{\mathcal{E}_{j,j+1}}=0,
	~\widehat{\mathcal{E}_{j,j+1}}\ast_{M}
	\mathcal{L}_{\widehat{\mathcal{B}_{j,j+1}}}
	\ast_{M}\mathcal{L}_{\widehat{\mathcal{Q}_{j,j+1}}}=0,\\
	&\mathcal{R}_{\widehat{\mathcal{A}_{j,j+1}}}\ast_{N}
	\widehat{\mathcal{E}_{j,j+1}}
	\ast_{M}\mathcal{L}_{\widehat{\mathcal{B}_{j,j+1}}}=0,
	~\mathcal{R}_{\widehat{\mathcal{C}_{j,j+1}}}\ast_{N}
	\widehat{\mathcal{E}_{j,j+1}}\ast_{M}\mathcal{L}_{\widehat{\mathcal{B}_{j,j+1}}}=0, ~(j=1,2),
\\
	&\mathcal{R}_{\widehat{\mathcal{M}}}\ast_{N}
	\mathcal{R}_{\widehat{\mathcal{A}}}\ast_{N}\widehat{\mathcal{E}}=0,
	~\widehat{\mathcal{E}}\ast_{M}\mathcal{L}_{\widehat{\mathcal{B}}}
	\ast_{M}\mathcal{L}_{\widehat{\mathcal{Q}}}=0,\\
	&\mathcal{R}_{\widehat{\mathcal{A}}}\ast_{N}
	\widehat{\mathcal{E}}
	\ast_{M}\mathcal{L}_{\widehat{\mathcal{B}}}=0,
	~\mathcal{R}_{\widehat{\mathcal{C}}}\ast_{N}
	\widehat{\mathcal{E}}\ast_{M}\mathcal{L}_{\widehat{\mathcal{B}}}=0,
	\end{align}
\end{enumerate}
Furthermore, if (\ref{item03}) holds, then the general solution to System (\ref{mainsystem02}) can be expressed as
\begin{align}
\label{july17equ054}
\mathcal{Z}_{k}&=\mathcal{A}_{k}^{\dagger}\ast_{N}\mathcal{E}_{k}\ast_{M}
\mathcal{B}_{k}^{\dagger}-\mathcal{A}_{k}^{\dagger}\ast_{N}
\mathcal{C}_{k}\ast_{N}
\mathcal{M}_{k}^{\dagger}\ast_{N}\mathcal{E}_{k}\ast_{M}
\mathcal{B}_{k}^{\dagger}
-\mathcal{A}_{k}^{\dagger}\ast_{N}\mathcal{S}_{k}\ast_{N}
\mathcal{C}_{k}^{\dagger}
\ast_{N}\mathcal{E}_{k}\ast_{M}\mathcal{N}_{k}^{\dagger}
\ast_{M}\nonumber\\
&\mathcal{D}_{k}\ast_{M}\mathcal{B}_{k}^{\dagger}-\mathcal{A}_{k}^{\dagger}
\ast_{N}\mathcal{S}_{k}\ast_{N}
\mathcal{W}_{k2}\ast_{M}\mathcal{R}_{\mathcal{N}_{k}}\ast_{M}
\mathcal{D}_{k}
\ast_{M}\mathcal{B}_{k}^{\dagger}+
\mathcal{L}_{\mathcal{A}_{k}}\ast_{N}\mathcal{W}_{k4}
+\mathcal{W}_{k5}\ast_{M}\mathcal{R}_{\mathcal{B}_{k}},&\\
\label{july17equ055}
\mathcal{Z}_{k+1}&=\mathcal{M}_{k}^{\dagger}\ast_{N}
\mathcal{E}_{k}\ast_{M}
\mathcal{D}_{k}^{\dagger}
+\mathcal{S}_{k}^{\dagger}\ast_{N}\mathcal{S}_{k}\ast_{N}
\mathcal{C}_{k}^{\dagger}
\ast_{N}\mathcal{E}_{k}\ast_{M}\mathcal{N}_{k}^{\dagger}
+\mathcal{L}_{\mathcal{M}_{k}}\ast_{N}\mathcal{L}_{\mathcal{S}_{k}}
\ast_{N}\mathcal{W}_{k1}
+\mathcal{L}_{\mathcal{M}_{k}}\ast_{N}\mathcal{W}_{k2}\nonumber\\
&\ast_{M}\mathcal{R}_{\mathcal{N}_{k}}+
\mathcal{W}_{k3}\ast_{M}\mathcal{R}_{\mathcal{D}_{k}},~(k=1,2,3,4)
\end{align}
where
\begin{align}
\mathcal{W}_{i1}&=\begin{bmatrix}
                    \mathcal{I}_{p_i}  &
                    0 \\
                  \end{bmatrix}\ast_{N}[
\widehat{\mathcal{A}_{i}}^{\dag}\ast_{N}
(\widehat{\mathcal{E}_{i}}-\widehat{\mathcal{C}_{i}}
\ast_{N}\mathcal{W}_{i2}\ast_{M}
\widehat{\mathcal{D}_{i}}-\widehat{\mathcal{F}_{i}}\ast_{N}
\mathcal{W}_{i+1,2}\ast_{M}
\widehat{\mathcal{G}_{i}})-\mathcal{T}_{i1}\ast_{M}
\widehat{\mathcal{B}_{i}}
+\mathcal{L}_{\widehat{\mathcal{A}_{i}}}\ast_{N}\mathcal{T}_{i2}],&\\
\mathcal{W}_{i+1,4}&=\begin{bmatrix}
                    0 &
                    \mathcal{I}_{p_i} \\
                  \end{bmatrix}\ast_{N}[
\widehat{\mathcal{A}_{i}}^{\dag}\ast_{N}
(\widehat{\mathcal{E}_{i}}-\widehat{\mathcal{C}_{i}}
\ast_{N}\mathcal{W}_{i2}\ast_{M}
\widehat{\mathcal{D}_{i}}-\widehat{\mathcal{F}_{i}}\ast_{N}
\mathcal{W}_{i+1,2}\ast_{M}
\widehat{\mathcal{G}_{i}})-\mathcal{T}_{i1}\ast_{M}
\widehat{\mathcal{B}_{i}}
+\mathcal{L}_{\widehat{\mathcal{A}_{i}}}\ast_{N}\mathcal{T}_{i2}],&\\
\mathcal{W}_{i3}&=[\mathcal{R}_{\widehat{\mathcal{A}_{i}}}\ast_{N}
(\widehat{\mathcal{E}_{i}}-\widehat{\mathcal{C}_{i}}\ast_{N}
\mathcal{W}_{i2}
\ast_{M}\widehat{\mathcal{D}_{i}}-
\widehat{\mathcal{F}_{i}}\ast_{N}\mathcal{W}_{i+1,2}\ast_{M}
\widehat{\mathcal{G}_{i}})
\ast_{M}\widehat{\mathcal{B}_{i}}^{\dag}
+\widehat{\mathcal{A}_{i}}\ast_{N}\mathcal{T}_{i1}&\\
&+\mathcal{T}_{i3}
\ast_{M}\mathcal{R}_{\widehat{\mathcal{B}_{i}}}]\ast_{M}
\begin{bmatrix}
  \mathcal{I}_{q_i} \\
   0 \\
\end{bmatrix},&\\
\mathcal{W}_{i+1,5}&=[\mathcal{R}_{\widehat{\mathcal{A}_{i}}}\ast_{N}
(\widehat{\mathcal{E}_{i}}-\widehat{\mathcal{C}_{i}}\ast_{N}
\mathcal{W}_{i2}
\ast_{M}\widehat{\mathcal{D}_{i}}-
\widehat{\mathcal{F}_{i}}\ast_{N}\mathcal{W}_{i+1,2}\ast_{M}
\widehat{\mathcal{G}_{i}})
\ast_{M}\widehat{\mathcal{B}_{i}}^{\dag}
+\widehat{\mathcal{A}_{i}}\ast_{N}\mathcal{T}_{i1}&\\
&+\mathcal{T}_{i3}
\ast_{M}\mathcal{R}_{\widehat{\mathcal{B}_{i}}}]\ast_{M}
\begin{bmatrix}
    0  \\
   \mathcal{I}_{q_i} \\
\end{bmatrix},&
\\
\label{july17equ056}
\mathcal{W}_{i2}&=\mathcal{A}_{ii}^{\dagger}\ast_{N}\mathcal{E}_{ii}\ast_{M}
\mathcal{B}_{ii}^{\dagger}-\mathcal{A}_{ii}^{\dagger}\ast_{N}\mathcal{C}_{ii}\ast_{N}
\mathcal{M}_{ii}^{\dagger}\ast_{N}\mathcal{E}_{ii}\ast_{M}\mathcal{B}_{ii}^{\dagger}
-\mathcal{A}_{ii}^{\dagger}\ast_{N}\mathcal{S}_{ii}\ast_{N}\mathcal{C}_{ii}^{\dagger}
\ast_{N}\mathcal{E}_{ii}\ast_{M}\mathcal{N}_{ii}^{\dagger}\ast_{M}\nonumber\\
&\mathcal{D}_{ii}\ast_{M}\mathcal{B}_{ii}^{\dagger}-\mathcal{A}_{ii}^{\dagger}
\ast_{N}\mathcal{S}_{ii}\ast_{N}
\mathcal{T}_{i4}\ast_{M}\mathcal{R}_{\mathcal{N}_{ii}}\ast_{M}\mathcal{D}_{ii}
\ast_{M}\mathcal{B}_{ii}^{\dagger}+
\mathcal{L}_{\mathcal{A}_{ii}}\ast_{N}\mathcal{T}_{i5}
+\mathcal{T}_{i6}\ast_{M}\mathcal{R}_{\mathcal{B}_{ii}},&\\
\label{july17equ057}
\mathcal{W}_{i+1,2}&=\mathcal{M}_{ii}^{\dagger}\ast_{N}
\mathcal{E}_{ii}\ast_{M}
\mathcal{D}_{ii}^{\dagger}
+\mathcal{S}_{ii}^{\dagger}\ast_{N}\mathcal{S}_{ii}\ast_{N}
\mathcal{C}_{ii}^{\dagger}
\ast_{N}\mathcal{E}_{ii}\ast_{M}\mathcal{N}_{ii}^{\dagger}
+\mathcal{L}_{\mathcal{M}_{ii}}\ast_{N}\mathcal{L}_{\mathcal{S}_{ii}}
\ast_{N}\mathcal{T}_{i7}
+\mathcal{L}_{\mathcal{M}_{ii}}\ast_{N}\mathcal{T}_{i4}\nonumber\\
&\ast_{M}\mathcal{R}_{\mathcal{N}_{ii}}+
\mathcal{T}_{i8}\ast_{M}\mathcal{R}_{\mathcal{D}_{ii}},~(i=1,2,3)&\\
\mathcal{T}_{j7}&=\begin{bmatrix}
                    \mathcal{I}_{m_j} & 0 \\
                  \end{bmatrix}\ast_{N}
[\mathcal{A}_{j,j+1}^{\dag}\ast_{N}(\mathcal{E}_{j,j+1}-\mathcal{C}_{j,j+1}
\ast_{N}\mathcal{T}_{j4}\ast_{M}
\mathcal{D}_{j,j+1}-\mathcal{F}_{j,j+1}\ast_{N}\mathcal{T}_{j+1,4}\ast_{M}
\mathcal{G}_{j,j+1})\nonumber\\
&-\mathcal{U}_{j1}\ast_{M}\mathcal{B}_{j,j+1}
+\mathcal{L}_{\mathcal{A}_{j,j+1}}\ast_{N}\mathcal{U}_{j2}],&\\
\mathcal{T}_{j+1,5}&=\begin{bmatrix}
                    0 & \mathcal{I}_{m_j} \\
                  \end{bmatrix}\ast_{N}
[\mathcal{A}_{j,j+1}^{\dag}\ast_{N}(\mathcal{E}_{j,j+1}-\mathcal{C}_{j,j+1}
\ast_{N}\mathcal{T}_{j4}\ast_{M}
\mathcal{D}_{j,j+1}-\mathcal{F}_{j,j+1}\ast_{N}\mathcal{T}_{j+1,4}\ast_{M}
\mathcal{G}_{j,j+1})\nonumber\\
&-\mathcal{U}_{j1}\ast_{M}\mathcal{B}_{j,j+1}
+\mathcal{L}_{\mathcal{A}_{j,j+1}}\ast_{N}\mathcal{U}_{j2}],&\\
\mathcal{T}_{j8}&=[\mathcal{R}_{\mathcal{A}_{j,j+1}}\ast_{N}
(\mathcal{E}_{j,j+1}-\mathcal{C}_{j,j+1}\ast_{N}\mathcal{U}_{12}
\ast_{M}\mathcal{D}_{j,j+1}-
\mathcal{F}_{j,j+1}\ast_{N}\mathcal{U}_{22}\ast_{M}\mathcal{G}_{j,j+1})
\ast_{M}\mathcal{B}_{j,j+1}^{\dag}\nonumber\\
&+\mathcal{A}_{j,j+1}\ast_{N}\mathcal{U}_{j1}+
\mathcal{U}_{j3}
\ast_{M}\mathcal{R}_{\mathcal{B}_{j,j+1}}]\ast_{M}
\begin{bmatrix}
\mathcal{I}_{n_j} \\
0 \\
\end{bmatrix},&
		\end{align}
\begin{align}
\mathcal{T}_{j+1,6}&=[\mathcal{R}_{\mathcal{A}_{j,j+1}}\ast_{N}
(\mathcal{E}_{j,j+1}-\mathcal{C}_{j,j+1}\ast_{N}\mathcal{U}_{12}
\ast_{M}\mathcal{D}_{j,j+1}-
\mathcal{F}_{j,j+1}\ast_{N}\mathcal{U}_{22}\ast_{M}\mathcal{G}_{j,j+1})
\ast_{M}\mathcal{B}_{j,j+1}^{\dag}\nonumber\\
&+\mathcal{A}_{j,j+1}\ast_{N}\mathcal{U}_{j1}+
\mathcal{U}_{j3}
\ast_{M}\mathcal{R}_{\mathcal{B}_{j,j+1}}]\ast_{M}
\begin{bmatrix}
0\\
\mathcal{I}_{n_j}  \\
\end{bmatrix},&\\
\mathcal{T}_{j4}&=\widehat{\mathcal{A}_{j,j+1}}^{\dagger}\ast_{N}
\widehat{\mathcal{E}_{j,j+1}}\ast_{M}\widehat{\mathcal{B}_{j,j+1}}^{\dagger}
-\widehat{\mathcal{A}_{j,j+1}}^{\dagger}\ast_{N}\widehat{\mathcal{C}_{j,j+1}}\ast_{N}
\widehat{\mathcal{M}_{j,j+1}}^{\dagger}\ast_{N}\widehat{\mathcal{E}_{j,j+1}}
\ast_{M}\mathcal{B}_{j,j+1}^{\dagger}-\nonumber\\
&\widehat{\mathcal{A}_{j,j+1}}^{\dagger}\ast_{N}\widehat{\mathcal{S}_{j,j+1}}
\ast_{N}\widehat{\mathcal{C}_{j,j+1}}^{\dagger}
\ast_{N}\widehat{\mathcal{E}_{j,j+1}}\ast_{M}
\widehat{\mathcal{N}_{j,j+1}}^{\dagger}\ast_{M}
\widehat{\mathcal{D}_{j,j+1}}\ast_{M}
\widehat{\mathcal{B}_{j,j+1}}^{\dagger}
-\widehat{\mathcal{A}_{j,j+1}}^{\dagger}\ast_{N}\nonumber\\
&\widehat{\mathcal{S}_{j,j+1}}\ast_{N}
\widehat{\mathcal{U}_{j4}}\ast_{M}\mathcal{R}_{\widehat{\mathcal{N}_{j,j+1}}}\ast_{M}
\widehat{\mathcal{D}_{j,j+1}}
\ast_{M}\widehat{\mathcal{B}_{j,j+1}}^{\dagger}
+\mathcal{L}_{\widehat{\mathcal{A}_{j,j+1}}}\ast_{N}\mathcal{U}_{j5}
+\mathcal{U}_{j6}\ast_{M}\mathcal{R}_{\widehat{\mathcal{B}_{j,j+1}}},&
\\
\mathcal{T}_{j+1,4}&=\widehat{\mathcal{M}_{j,j+1}}^{\dag}\ast_{N}
\widehat{\mathcal{E}_{j,j+1}}\ast_{N}\widehat{\mathcal{D}_{j,j+1}}^{\dag}
+\widehat{\mathcal{S}_{j,j+1}}^{\dag}\ast_{N}\widehat{\mathcal{S}_{j,j+1}}
\ast_{N}\widehat{\mathcal{C}_{j,j+1}}\ast_{N}
\widehat{\mathcal{E}_{j,j+1}}\ast_{M}
\widehat{\mathcal{N}_{j,j+1}}^{\dag}+\nonumber\\
&\mathcal{L}_{\widehat{\mathcal{M}}_{j,j+1}}\ast_{N}
\mathcal{L}_{\widehat{\mathcal{S}}_{j,j+1}}
\ast_{N}\mathcal{U}_{j7}
+\mathcal{L}_{\widehat{\mathcal{M}_{j,j+1}}}\ast_{N}
\mathcal{U}_{j4}\ast_{M}\mathcal{R}_{\widehat{\mathcal{N}_{j,j+1}}}+
\mathcal{U}_{j8}\ast_{M}\mathcal{R}_{\widehat{\mathcal{D}_{j,j+1}}},
~(j=1,2)&\\
\mathcal{U}_{17}&=\begin{bmatrix}
                    \mathcal{I}_s & 0 \\
                  \end{bmatrix}\ast_{N}
[\mathcal{A}^{\dag}\ast_{N}(\mathcal{E}-\mathcal{C}
\ast_{N}\mathcal{U}_{14}\ast_{M}
\mathcal{D}-\mathcal{F}\ast_{N}\mathcal{U}_{24}\ast_{M}
\mathcal{G})
-\mathcal{V}_{1}\ast_{M}\mathcal{B}
+\mathcal{L}_{\mathcal{A}}\ast_{N}\mathcal{V}_{2}],&\\
\mathcal{U}_{25}&=\begin{bmatrix}
                    0 & \mathcal{I}_s \\
                  \end{bmatrix}\ast_{N}
[\mathcal{A}^{\dag}\ast_{N}(\mathcal{E}-\mathcal{C}
\ast_{N}\mathcal{U}_{14}\ast_{M}
\mathcal{D}-\mathcal{F}\ast_{N}\mathcal{U}_{24}\ast_{M}
\mathcal{G})
-\mathcal{V}_{1}\ast_{M}\mathcal{B}
+\mathcal{L}_{\mathcal{A}}\ast_{N}\mathcal{V}_{3}],&\\
\mathcal{U}_{18}&=[\mathcal{R}_{\mathcal{A}}\ast_{N}
(\mathcal{E}-\mathcal{C}\ast_{N}\mathcal{U}_{14}
\ast_{M}\mathcal{D}-
\mathcal{F}\ast_{N}\mathcal{U}_{24}\ast_{M}\mathcal{G})
\ast_{M}\mathcal{B}^{\dag}
+\mathcal{A}\ast_{N}\mathcal{U}_{j1}+
\mathcal{U}_{j3}
\ast_{M}\mathcal{R}_{\mathcal{B}}]\nonumber\\
&\ast_{M}
\begin{bmatrix}
\mathcal{I}_t \\
0 \\
\end{bmatrix},&\\
\mathcal{U}_{26}&=[\mathcal{R}_{\mathcal{A}}\ast_{N}
(\mathcal{E}-\mathcal{C}\ast_{N}\mathcal{U}_{14}
\ast_{M}\mathcal{D}-
\mathcal{F}\ast_{N}\mathcal{U}_{24}\ast_{M}\mathcal{G})
\ast_{M}\mathcal{B}^{\dag}
+\mathcal{A}\ast_{N}\mathcal{U}_{j1}+
\mathcal{U}_{j3}
\ast_{M}\mathcal{R}_{\mathcal{B}}]\nonumber\\
&\ast_{M}
\begin{bmatrix}
0\\
\mathcal{I}_t  \\
\end{bmatrix},&\\
\mathcal{U}_{14}&=\widehat{\mathcal{A}}^{\dagger}\ast_{N}
\widehat{\mathcal{E}}\ast_{M}\widehat{\mathcal{B}}^{\dagger}
-\widehat{\mathcal{A}}^{\dagger}\ast_{N}\widehat{\mathcal{C}}\ast_{N}
\widehat{\mathcal{M}}^{\dagger}\ast_{N}\widehat{\mathcal{E}}
\ast_{M}\mathcal{B}^{\dagger}-
\widehat{\mathcal{A}}^{\dagger}\ast_{N}\widehat{\mathcal{S}}
\ast_{N}\widehat{\mathcal{C}}^{\dagger}
\ast_{N}\widehat{\mathcal{E}}\ast_{M}
\widehat{\mathcal{N}}^{\dagger}\ast_{M}
\widehat{\mathcal{D}}\ast_{M}
\widehat{\mathcal{B}}^{\dagger}\nonumber\\
&-\widehat{\mathcal{A}}^{\dagger}\ast_{N}
\widehat{\mathcal{S}}\ast_{N}
\mathcal{V}_{4}\ast_{M}
\mathcal{R}_{\widehat{\mathcal{N}}}\ast_{M}
\widehat{\mathcal{D}}
\ast_{M}\widehat{\mathcal{B}}^{\dagger}
+\mathcal{L}_{\widehat{\mathcal{A}}}\ast_{N}\mathcal{V}_{5}
+\mathcal{V}_{6}\ast_{M}\mathcal{R}_{\widehat{\mathcal{B}}},&\\
\mathcal{U}_{24}&=\widehat{\mathcal{M}}^{\dag}\ast_{N}
\widehat{\mathcal{E}}\ast_{N}\widehat{\mathcal{D}}^{\dag}
+\widehat{\mathcal{S}}^{\dag}\ast_{N}\widehat{\mathcal{S}}
\ast_{N}\widehat{\mathcal{C}}\ast_{N}
\widehat{\mathcal{E}}\ast_{M}
\widehat{\mathcal{N}}^{\dag}+
\mathcal{L}_{\widehat{\mathcal{M}}}\ast_{N}
\mathcal{L}_{\widehat{\mathcal{S}}}
\ast_{N}\mathcal{V}_{7}
+\mathcal{L}_{\widehat{\mathcal{M}}}\ast_{N}
\mathcal{V}_{4}\ast_{M}
\mathcal{R}_{\widehat{\mathcal{N}}}\nonumber\\
&+\mathcal{V}_{8}\ast_{M}\mathcal{R}_{\widehat{\mathcal{D}}},
\end{align}
and $\mathcal{T}_{i1},~\mathcal{T}_{i2},~\mathcal{T}_{i3},~
~\mathcal{U}_{j1},~\mathcal{U}_{j2},~\mathcal{U}_{j3},~
~\mathcal{T}_{14},~\mathcal{T}_{15},~
\mathcal{T}_{16},~\mathcal{T}_{37},~\mathcal{T}_{38},~ \mathcal{U}_{15},~\mathcal{U}_{16},~~\mathcal{U}_{27},$
$~\mathcal{U}_{28},
~\mathcal{V}_{t},~t=1,\ldots,8$ are arbitrary tensors with appropriate sizes over $\mathbb{H}$. $p_i,m_j,s$ is the same as the row block of $\mathcal{C}_i,\mathcal{C}_{j+1},\mathcal{C}_3$, respectively. $q_i,n_j,t$ is the same as the column block of $\mathcal{D}_i,\mathcal{D}_{j+1},\mathcal{D}_2$, $(i=1,2,3;~j=1,2)$.
\end{theorem}
\begin{proof}
(\ref{item03})$\Leftrightarrow$(\ref{item04}): We divide the proof into two parts:
\begin{itemize}
  \item Firstly, we separate the system (\ref{mainsystem02}) into four equations as following
\begin{align}
&\label{july17equ058}\mathcal{A}_1\ast_{M}\mathcal{Z}_1\ast_{N}\mathcal{B}_1+
\mathcal{C}_1\ast_{M}\mathcal{Z}_{2}\ast_{N}\mathcal{D}_1
=\mathcal{E}_1,\\
&\label{july17equ059}\mathcal{A}_2\ast_{M}\mathcal{Z}_2\ast_{N}\mathcal{B}_2+
\mathcal{C}_2\ast_{M}
\mathcal{Z}_{3}\ast_{N}\mathcal{D}_2=\mathcal{E}_2,
\end{align}
\begin{align}
&\label{july17equ060}\mathcal{A}_3\ast_{M}\mathcal{Z}_3\ast_{N}\mathcal{B}_3+
\mathcal{C}_3\ast_{M}
\mathcal{Z}_{4}\ast_{N}\mathcal{D}_3=\mathcal{E}_3,\\
&\label{july17equ061}\mathcal{A}_4\ast_{M}\mathcal{Z}_4\ast_{N}\mathcal{B}_4+
\mathcal{C}_4\ast_{M}
\mathcal{Z}_{5}\ast_{N}\mathcal{D}_4=\mathcal{E}_4.
\end{align}
Using four of Lemma \ref{lemma1},
it can be easily proved that the  solvable conditions (\ref{july17equ052}), (\ref{july17equ053}) hold when the system is consistent, and the expression of general solution are given by (\ref{july17equ054}) and (\ref{july17equ055}), (k=1,2,3,4).
  \item Let $\mathcal{Z}_2$ in Equation (\ref{july17equ058}) be equal to $\mathcal{Z}_2$ in Equation (\ref{july17equ059}), $\mathcal{Z}_3$ in Equation (\ref{july17equ059}) be equal to $\mathcal{Z}_3$ in Equation (\ref{july17equ060}) and $\mathcal{Z}_4$ in Equation (\ref{july17equ060})) be equal to $\mathcal{Z}_4$ in Equation (\ref{july17equ061}). Then we obtain a new tensor system as following
\begin{equation}
\begin{cases}
\widehat{\mathcal{A}_1}\ast_{N}\begin{bmatrix}
                                 W_{11} \\
                                 W_{24} \\
                               \end{bmatrix}
+
\begin{bmatrix}
  W_{13} &  W_{25} \\
\end{bmatrix}
\ast_{M}\widehat{\mathcal{B}_1}+
\widehat{\mathcal{C}_1}\ast_{N}
\mathcal{W}_{12}\ast_{M}\widehat{\mathcal{D}_1}+
\widehat{\mathcal{F}_1}\ast_{N}
\mathcal{W}_{22}\ast_{M}\widehat{\mathcal{G}_1}=
\widehat{\mathcal{E}_1}&\\
\widehat{\mathcal{A}_2}\ast_{N}\begin{bmatrix}
                                 W_{21} \\
                                 W_{34} \\
                               \end{bmatrix}
+
\begin{bmatrix}
  W_{23} &  W_{35} \\
\end{bmatrix}
\ast_{M}\widehat{\mathcal{B}_2}+
\widehat{\mathcal{C}_2}\ast_{N}
\mathcal{W}_{22}\ast_{M}\widehat{\mathcal{D}_2}+
\widehat{\mathcal{F}_2}\ast_{N}
\mathcal{W}_{32}\ast_{M}\widehat{\mathcal{G}_2}=
\widehat{\mathcal{E}_2}&\\
\widehat{\mathcal{A}_3}\ast_{N}\begin{bmatrix}
                                 W_{31} \\
                                 W_{44} \\
                               \end{bmatrix}
+
\begin{bmatrix}
  W_{33} &  W_{45} \\
\end{bmatrix}
\ast_{M}\widehat{\mathcal{B}_3}+
\widehat{\mathcal{C}_3}\ast_{N}
\mathcal{W}_{32}\ast_{M}\widehat{\mathcal{D}_3}+
\widehat{\mathcal{F}_3}\ast_{N}
\mathcal{W}_{42}\ast_{M}\widehat{\mathcal{G}_3}=
\widehat{\mathcal{E}_3}.&\\
\end{cases}
\end{equation}
Obviously, the necessary and sufficient conditions and general solution can be given by Theorem \ref{thm1}. Therefore, more details are omitted here.
\end{itemize}
\end{proof}

\section{Conclusion and Further Work}
We have obtained a necessary and sufficient condition for the existence of the general solution to (\ref{mainsystem01}) via Einstein product by using M-P inverse in Theorem \ref{thm1}. We also have presented an expression of  the general solution to (\ref{mainsystem01})) when it is solvable.  Moreover, the general solution to tensor equation (\ref{mainsystem02}) has been considered. Some known results can be viewed as special cases of the one obtained in this paper.

Based on the above conclusion, we next consider the generalization of tensor system:
\[\mathcal{A}_i\ast_{N}\mathcal{X}_i+
\mathcal{Y}_i\ast_{M}\mathcal{B}_i+\mathcal{C}_i\ast_{N}
\mathcal{Z}_i\ast_{M}\mathcal{D}_i+\mathcal{F}_i\ast_{N}
\mathcal{Z}_{i+1}\ast_{M}\mathcal{G}_i=\mathcal{E}_i, i=1,2,\ldots,n,\]
giving the solvable conditions and presenting the general solution. This work is more meaningful, but we need some new methods to solve it.

\end{document}